\documentclass[conference]{IEEEtran}
\usepackage{graphicx}
 \usepackage{graphicx}
 \usepackage{mathtools}

\usepackage{relsize}
\usepackage{caption}

\usepackage{amsmath}
\usepackage{amsthm}
\newtheoremstyle{exampstyle}
{3pt} % Space above
{3pt} % Space below
{\itshape} % Body font
{} % Indent amount
{\bfseries}
{.} % Punctuation after theorem head
{.5em} % Space after theorem head
{} % Theorem head spec (can be left empty, meaning `normal')
\theoremstyle{exampstyle}

\usepackage{amssymb}
\usepackage{extpfeil}
\usepackage{extarrows}
\usepackage{mdsymbol}
\usepackage{multirow}
\usepackage{color, colortbl}
\usepackage{booktabs,subcaption,amsfonts,dcolumn}

\DeclareMathAlphabet{\mathcal}{OMS}{cmsy}{m}{n}

\newcommand\equi{\overset{\hspace{5pt}\hour^{k}}{=\joinrel=}}
\newtheorem{theorem}{Theorem}
\newtheorem{example}{Example}

\newtheorem{definition}{Definition}
\newtheorem{lemma}{Lemma}
\newtheorem{cor}{Corollary}
\usepackage[linesnumbered,ruled, noend]{algorithm2e}
 \usepackage{algpseudocode}
 \usepackage{tikz}
 \usetikzlibrary{shapes.geometric}
 	\definecolor{red}{HTML}{FF0000}
 	\definecolor{green}{HTML}{33FF33}
 	\definecolor{blue}{HTML}{3333FF}
 	\definecolor{grey}{HTML}{A9A9A9}
 	
 \tikzset{
  graph vertex/.style={
    circle,
    draw,
  },
  graph directed edge/.style={
    ->,
    >=stealth,
    thick,
  },
  graph tree edge/.style={
    graph directed edge
  },
  graph forward edge/.style={
    graph directed edge,
    every edge/.style={
      edge node={node [fill=white,font=\scriptsize] {f}},
      loosely dotted,
      draw,
    },
  },
  graph back edge/.style={
    graph directed edge,
    every edge/.style={
      edge node={node [fill=white,font=\scriptsize] {0}},
      draw,
    },
  },
  graph back edge1/.style={
    graph directed edge,
    every edge/.style={
      edge node={node [fill=white,font=\scriptsize] {1}},
      draw,
    },
  },
  graph back edge2/.style={
    graph directed edge,
    every edge/.style={
      edge node={node [fill=white,font=\scriptsize] {2}},
      draw,
    },
  },
  graph back edge3/.style={
    graph directed edge,
    every edge/.style={
      edge node={node [fill=white,font=\scriptsize] {3}},
      draw,
    },
  },
  graph cross edge/.style={
    graph directed edge,
    every edge/.style={
      edge node={node [fill=white,font=\scriptsize] {c}},
      dotted,
      draw,
    },
  },
}
\SetAlgoNoEnd %\SetAlgoSkip
\newcommand{\hour}{\mathbin{\rotatebox[origin=c]{90}{$\hourglass$}}}
\begin{document}

% ****************** TITLE ****************************************

%\title{Efficient Index-based Approaches for Personalized $k$-wing Search in Large and Dynamic Bipartite Graphs}
\title{Searching Personalized $k$-wing in Large and Dynamic Bipartite Graphs}

%\author{
%}

\author{
Aman Abidi, %\textsuperscript{$\dagger$}, 
Lu Chen, %\textsuperscript{$\dagger$}, 
Rui Zhou, %\thanks{\textsuperscript{$\ddagger$} This work is done while the author is affiliated with Swinburne University of Technology. }\textsuperscript{$\ddagger$}, 
Chengfei Liu\\ 
%Li\thanks{\textsuperscript{$\P$} This work is done while the author is affiliated with The University of Western Australia. }\textsuperscript{$\P$}, Rui Zhou\textsuperscript{$\dagger$}\\

%, , and Rui Zhou\textsuperscript{$\ddagger$}\\
%\affiliation{
%\textsuperscript{$\dagger$}
{Swinburne University of Technology}\\ %\textsuperscript{$\ddagger$}{Charles Darwin University}, \textsuperscript{$\P$}{Deakin University}\\
%\textsuperscript{$\P$}{University of Western Australia}\\
%\textsuperscript{$\dagger$}
\{aabidi, luchen, cliu, rzhou\}@swin.edu.au
%\textsuperscript{$\ddagger$}Kewen.liao@cdu.edu.au
%\textsuperscript{$\P$}jianxin.li@deakin.edu.au
}

\maketitle

\SetAlFnt{\small\normalfont}
\SetAlCapHSkip{0em}

\everymath{\small}

\setlength\floatsep{1.25\baselineskip plus 3pt minus 3pt}

\setlength\textfloatsep{1.25\baselineskip plus 3pt minus 3pt}
\setlength\intextsep{1.25\baselineskip plus 3pt minus 3 pt}

\setlength{\abovecaptionskip}{5pt}
\setlength{\belowcaptionskip}{-2pt}

\setlength{\abovedisplayskip}{8pt}
\setlength{\belowdisplayskip}{8pt}

%\renewtheorem{definition}{Definition}
%\renewtheorem{theorem}{Theorem}
%\renewtheorem{lemma}{Lemma}
\begin{abstract}
%no change to 1st sentence
%Discovering all the bipartite cohesive subgraphs in a bipartite graph has found extensive applications. 
There are extensive studies focusing on the application scenario that all the bipartite cohesive subgraphs need to be discovered in a bipartite graph. 
However, we observe that, for some applications, one is interested in finding bipartite cohesive subgraphs containing a specific vertex. In this paper, we study a new query dependent bipartite cohesive subgraph search problem based on $k$-wing model, named as the personalized $k$-wing search problem. We introduce a $k$-wing equivalence relationship to summarize the edges of a bipartite graph $G$ into groups. Therefore, all the edges of $G$ are segregated into different groups, i.e. \textit{$k$-wing equivalence class}, forming an efficient and wing number conserving index called \textit{EquiWing}. Further, we propose a more compact version of \textit{EquiWing}, \textit{EquiWing-Comp}, which is achieved by integrating our proposed \textit{$k$-butterfly loose} approach and discovered hierarchy properties. These indices are used to expedite the personalized $k$-wing search with a non-repetitive access to $G$, which leads to linear algorithms for searching the personalized $k$-wing. 
Moreover, we conduct a thorough study on the maintenance of the proposed indices for evolving bipartite graphs. We discover novel properties that help us localize the scope of the maintenance at a low cost. 
By exploiting the discoveries, we propose novel algorithms for maintaining the two indices, which substantially reduces the cost of maintenance. 
We perform extensive experimental studies in real, large-scale graphs to validate the efficiency and effectiveness of \textit{EquiWing} and \textit{EquiWing-Comp} compared to the baseline.  %has achieved at least an order of magnitude speedup in cohesive subgraph search over \textit{EquiWing}. 
\end{abstract}

%\pagestyle{empty}
%The existing approaches take quadratic time to compute a $k$-wing and need to explore the entire bipartite graph. 
\section{Introduction}
\noindent\textbf{Bipartite cohesive subgraph}. A bipartite graph is an interesting structure that can be used to represent large heterogeneous data. Many real-world networks can be modeled using bipartite graphs such as customer-product, product-rating and author-paper (Fig. \ref{fig:my_label_motivation}) networks. For a given bipartite graph $G=(U,V,E)$, a bipartite cohesive subgraph $B=(U',V',E') \subseteq G$ such that $U'$ and $V'$ are extensively connected via edges in $E'$. A bipartite cohesive subgraph $B=(U',V',E')$ is unaffected by the connections within $U'$ or $V'$ themselves. 
%the collection of densely connected components of $G$, s.t. $U' \subseteq U, V \subseteq V$, and $E' \subseteq E$. 
%For instance, given a bipartite graph with a set of authors and a set of papers where if 
%For example, two authors working in the same area can be in each other's \textit{bipartite cohesive subgraph} without knowing each other. 
%In Fig. \ref{fig:my_label_motivation}, the authors $v_0$ and $v_3$ might not know each other, yet they can be recommended to each other for potential future work. 

\noindent\textit{Applications}. Finding bipartite cohesive subgraphs has a rich literature with examples such as word and document clustering \cite{Dhillon:2001:CDW:502512.502550}, spam group detection in the web \cite{Gibson:2005:DLD:1083592.1083676}, and sponsored advertisement \cite{fain2006sponsored}. The objective of finding all the bipartite cohesive subgraphs %i.e. \textit{bipartite cohesive subgraph detection} 
is to aim the entire bipartite graph and generally apply a global criterion to provide macroscopic information.% about the graph. 

\noindent\textit{Existing models}. Due to a large number of applications, % for  in the real-world, %various concept of bipartite cohesiveness has been proposed to find all the cohesive subgraphs in a bipartite graph \cite{ DBLP:conf/wsdm/SariyuceP18,zou2016bitruss, DBLP:conf/icde/Wang0Q0020}.
various bipartite cohesive subgraph models have been studied such as bicliques \cite{abidipivot}, $(\alpha,\beta)$-core \cite{liu2020efficient}.
Recently, \textit{$k$-bitruss} \cite{ DBLP:conf/wsdm/SariyuceP18} and \textit{$k$-wing} \cite{zou2016bitruss} have drawn great attention since they exhibit high cohesiveness while can be discovered in polynomial time. %These three structures correspond to cliques, k-core, and k-truss in the general graphs. 
\textit{Bitruss subgraph} is defined based on butterfly (i.e., a complete $2\times 2$ bipartite graph). A bipartite subgraph is a $k$-bitruss if it first satisfies the minimum $k$-butterfly constraint defined as every edge in the subgraph participates in at least $k$ butterflies and then it is non-extensible. 
%TODO1:fixed
A $k$-wing further strengthens a $k$-bitruss by adding an extra constraint after the minimum $k$-butterfly constraint of the $k$-bitruss. %A $k$-bitruss is a $k$-wing if it additionally satisfies the 
The extra constraint is called butterfly connectivity constraint, defined as any two edges in the subgraph are either in a butterfly or can be reached by a set of edge-overlapping butterflies. 
%introduced in \cite{zou2016bitruss} proposed a new notion of a dense subgraph which uses butterfly as a motif. A bipartite subgraph is a $k$-bitruss if every edge in the subgraph participates in at least $k$ butterflies. 
%of a bipartite graph wherein each edge of the subgraph has to participate in at least $k$ rectangles (\textit{bitruss number}). 
%Further, \textit{$k$-wing subgraph} \cite{DBLP:conf/wsdm/SariyuceP18} was also presented using the same motif but with an extra notion of connectivity via butterflies among the edges of a \textit{$k$-wing} subgraph. 
Although the edges of a $k$-wing (\cite{DBLP:conf/wsdm/SariyuceP18,zou2016bitruss}) exhibit the same cohesiveness as a $k$-bitruss, the \textit{$k$-wing} differentiates the $k$-bitruss using the butterfly connectivity constraint, which makes a $k$-wing superior. 
%8/10/20, updateing the example
% The bipartite graph in Fig. \ref{fig:my_label_motivation} contains a single \textit{$3$-bitruss} (whole graph) and two \textit{$1$-wings} (dashed lines). 
% Let Fig. \ref{fig:my_label_motivation} represent an \textit{author-paper network} from the domain of \textit{Computer Science}. A \textit{bipartite cohesive subgraph} would represent the authors which can potentially work together. We can observe that a \textit{$k$-wing} can further differentiate more cohesive subgraphs within a the \textit{$k$-bitruss}, i.e. edges share the subdomain are grouped together. 
The bipartite graph in Fig. \ref{fig:my_label_motivation} contains a single \textit{$3$-bitruss} (blue) and two \textit{$3$-wings} (red and orange), consisting of edges represented with solid lines. 
Let Fig.~\ref{fig:my_label_motivation} represent an \textit{author-paper network} from the domain of \textit{Computer Science}. Finding \textit{bipartite cohesive subgraph} would allow us to discover the authors who can potentially work together. 
%TODO2:fixed
We can observe that the \textit{$k$-wing} model can further discover more cohesive subgraphs (\textit{Graph Data Management (GDM (\textit{$3$-wing}))} and \textit{Spatial Data Management (SDM (\textit{$3$-wing}))}) compared with the $3$-bitruss (\textit{Database (DB)}) discovered by the $k$-bitruss model. 

\begin{figure}[]
\centering
\scalebox{0.2}{

\begin{tikzpicture}

\node(p)[] at (4.5,9.5) { \Huge Computer Science (author-paper network)};
\draw [,] (-2.6,0) -- (40.6,0);
\draw [,] (40.6,0) -- (40.6,9);
\draw [,] (40.6,9) -- (-2.6,9);
\draw [,] (-2.6,9) -- (-2.6,0);
%\node(p)[blue] at (14,9.1) { \Huge AI };
\node(p)[blue] at (-0.2,8.6) { \Huge DB };
\node(p)[blue] at (-0.2,7.8) { \Huge ($3$-bitruss)};
\node(p)[orange] at (30.45,7.7) { \Huge SDM };
\node(p)[orange] at (30.45,7) { \Huge ($3$-wing)};
\node(p)[red] at (13.25,7.5) { \Huge GDM };
\node(p)[red] at (13.25,6.8) { \Huge ($3$-wing)};
%\node(p)[] at (-4,6) { \Huge $Authors$};
%\node(p)[] at (-4,2) { \Huge $Papers$};
\draw [,blue] (7.8,0.3) -- (40.3,0.3);
\draw [,blue] (40.3,0.3) -- (40.3,8.7);
\draw [,blue] (40.3,8.7) -- (1.7,8.7);
\draw [,blue] (7.8,0.3) -- (7.8,4.3);
\draw [,blue] (7.8,4.3) -- (1.7,4.3);
\draw [,blue] (1.7,4.3) -- (1.7,8.7);

%\draw[dashed,blue] (7.5,0) .. controls (5.5,6) and (-1,1) .. (2.5,8.2);

\draw [, orange] (20.5,0.6) -- (40,0.6);
\draw [, orange] (40,0.6) -- (40,8.4);
\draw [, orange] (40,8.4) -- (20.5,8.4);
\draw [, orange] (20.5,8.4) -- (20.5,0.6);

\draw [,red] (8.1,0.9) -- (23,0.9);
\draw [,red] (23,0.9) -- (23,8.1);
\draw [,red] (23,8.1) -- (2,8.1);
\draw [,red] (8.1,0.9) -- (8.1,4.6);
\draw [,red] (8.1,4.6) -- (2,4.6);
\draw [,red] (2,4.6) -- (2,8.1);
%\draw[dashed,red] (7.7,0) .. controls (5.7,6.2) and (-0.8,1.2) .. (3,8);

\node (v_8) at ( 39,6) [] {\includegraphics[scale=0.05]{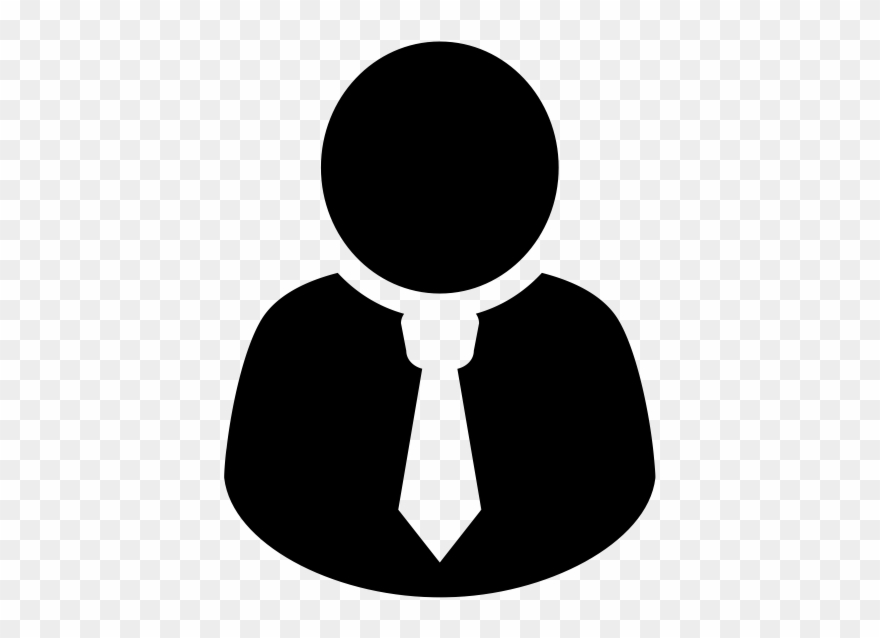} };
\node () at (39,7) [] {\Huge $v_8$};

\foreach \name/ \x in {v_1/-.7,v_2/5,v_3/10.4,v_4/16.1,v_5/21.85,v_6/27.6,v_7/33.3}
\node (\name) at ( \x,6) [] {\includegraphics[scale=0.05]{author_clipart.png} };
\foreach \name/ \x in {v_1/-1,v_2/4.7,v_3/10.4,v_4/16.1,v_5/21.85,v_6/27.6,v_7/33.3}
\node () at (\x,7) [] {\Huge $\name$};

% \foreach \name/ \x in {v_0/2,v_1/6,v_2/10,v_4/18,v_5/22}
% \node () at (\x,7) [] {\Huge $\name$};

\foreach \name/ \x in {u_1/2,u_2/7,u_3/13.5,u_4/18.975,u_5/24.725,u_6/30.45,u_7/36.15}
\node (\name) at ( \x,2.3) [] {\includegraphics[scale=0.07]{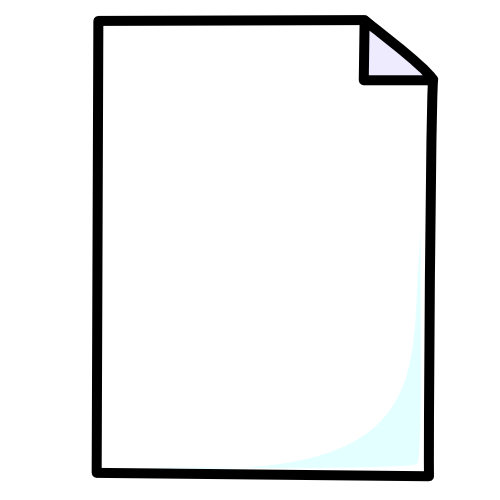}};

\foreach \name/ \x in {u_1/2,u_2/7,u_3/13.25,u_4/18.975,u_5/24.725,u_6/30.45,u_7/36.15}
\node () at ( \x,1.3) [] {\Huge $\name$};
\path[]
(v_1) 
edge  [very thick,loosely dashed](u_1)
edge  [very thick,loosely dashed](u_2)
(v_2) edge  [very thick,loosely dashed] (u_1)
edge [very thick, loosely dashed] (u_2)
edge  [very thick,] (u_3) 
edge  [very thick,] (u_4)
(v_3) edge [very thick,] (u_4)
edge  [very thick,loosely dashed](u_2)
edge  [very thick,](u_3)
(v_4) 
edge  [very thick,](u_4)
edge  [very thick,](u_3)
(v_5) edge  [very thick,](u_3)
edge  [very thick,](u_4)
edge  [very thick,](u_5)
edge  [very thick,](u_6)
(v_6) edge  [very thick,](u_5)
edge  [very thick,](u_6)
edge  [very thick,](u_7)
edge  [very thick,loosely dashed](u_4)
(v_7) edge  [very thick,](u_5)
edge  [very thick,](u_6)
edge  [very thick,](u_7)
(v_8) edge  [very thick,](u_5)
edge  [very thick,](u_6)
edge  [very thick,](u_7);
\end{tikzpicture}
}
\caption{$k$-bitruss, $k$-wing, and personalized $k$-wing}
\label{fig:my_label_motivation}
\vspace{-10pt}
\end{figure} 

\noindent\textbf{Personalized bipartite cohesive subgraph}. Distinct vertices in a bipartite graph may have distinct properties, which requires the microscopic analysis, i.e. personalized. There are many real-world applications where people are more interested in personalized bipartite cohesive subgraphs rather than all possible bipartite cohesive subgraphs, 
%For example, analysis in co-authorship networks (e.g. DBLP) for improving the collaboration among the authors for potential future collaborations and project fundings \cite{arnab2016analysis,Sozio2010cocktail}. 
%TODO3: fixed
e.g., analysis in co-authorship networks (e.g. DBLP) for improving the author's potential future collaborations and project fundings \cite{arnab2016analysis,Sozio2010cocktail}. 
To provide a better insight, we demonstrate how personalized bipartite cohesive subgraph using \textit{$k$-wing} model can help authors $v_8$ and $v_5$ discover their potential future collaborators in Fig. \ref{fig:my_label_motivation}. %We can observe that a \textit{$3$-bitruss} for $v_8$ recommends authors from both the subdomains \textit{ML} and \textit{DL} with equal significance for its future collaboration, which is highly unlikely. However, 
%seems to be
The \textit{$3$-wing} including $v_8$, recommends the 
%TODO4:fixed 
authors from the \textit{SDM} subdomain in which $v_8$ works, i.e. $v_5,v_6$ and $v_7$, which is more practical and promising in comparison to recommend the authors from
%TODO5:fixed
\textit{GDM} ($v_2,v_3$, and $v_4$) as $v_8$ is 
%TODO6:fixed
inexperienced in \textit{GDM}. Also, we do not need to enumerate all \textit{$3$-wings} as it may produce some uninteresting results where $v_8$ is not participating 
%TODO7:fixed
i.e. \textit{GDM}. %Therefore, we only enumerate the personalized \textit{$3$-wings} for $v_8$ i.e. \textit{DL}. 
For the author $v_5$, the personalized cohesive bipartite subgraph consists of two \textit{$3$-wings} % and \textit{$3$-bitruss} is the same 
suggesting authors from two subdomains i.e. $v_5$ can 
%TODO8:fixed
collaborate with the authors either from \textit{GDM} or \textit{SDM} with same priority as $v_5$ is experienced in %TODO9:fixed
both \textit{GDM} and \textit{SDM}. 
In contrast, if we omit the personalization, the search for authors $v_{8}$ and $v_{5}$ would get the same result consisting of all the $3$-wings in Fig.~\ref{fig:my_label_motivation}.  
The problem of personalized bipartite cohesive subgraph search can be used for a wide variety of other applications such as personalized recommendation for products \cite{zhang2019domain} and hotels \cite{kaya2019hotel}, identifying potential websites for banner advertisements \cite{hunter2013structural}, efficient training of the employees on the internal projects by choosing the optimal team for an employee \cite{zhang2017enterprise} and speculating the drug-target interactions \cite{yamanishi2013chemogenomic}.

In this paper, we study \textit{personalized $k$-wing search} based on the $k$-wing model. 
Given a bipartite graph $G(U,V,E)$, a query vertex $q\in U\cup V$ and an integer $k$, the personalized $k$-wing search returns all the \textit{$k$-wings}. % containing $q$. 
%The selection of \textit{$k$-wing} is based on its possession of the features of bicliques (high density) and $(\alpha,\beta)$-core (hierarchical properties and polynomial-time efficiency) simultaneously.
% The selection of \textit{$k$-wing} is based on the following \textit{$k$-wing} properties: (1) Discovering all \textit{$k$-wings} in a bipartite graph requires polynomial time \cite{DBLP:conf/wsdm/SariyuceP18}. As a result, it has the potential to be extensively used for real-world applications; (2) The motif butterfly for the \textit{$k$-wing} which acts as its building block, is a stable and strong relationship. It also ensures the high-density and cohesiveness of the subgraphs. 
The problem of personalized $k$-wing search can be addressed by adapting the algorithm proposed for finding all the \textit{$k$-wings} \cite{DBLP:conf/wsdm/SariyuceP18}, which serves as our baseline. 
The idea is, we first compute the $k$-bitrusses using the state-of-the-art bitruss algorithm \cite{DBLP:conf/icde/Wang0Q0020}. 
Then, among the $k$-bitrusses, we further explore within the $k$-bitruss containing $q$ to form the personalized 
%%TODO9:fixed
\textit{$k$-wings} by grouping edges in the the $k$-bitruss containing $q$ together if the edges are butterfly connected.  
However, even if we ignore the cost for computing the 
%TODO10:fixed %TODO11:fixed
$k$-bitruss, the baseline still has high time complexity, i.e., $\mathcal{O}(d^{2}_{max}|E|)$, where $d_{max}$ is the maximum degree of $G$ and $|E|$ is the number of edges in the personalized \textit{$k$-wings}.

\begin{table}[t]
    \centering
 \caption{Notations and their descriptions}
\scalebox{.75}{
\centering
\begin{tabular}[t]{|m{9em}|m{25em}|}
\hline
\rowcolor{grey}
Notation & Description\\
\hline
$G=(U,V,E)$ & A simple bipartite graph $G$%, where $U$ and $V$ are the two disjoint sets, and $E$ is the set of Edges.
\\
\hline
$EW(\raisebox{1.5pt}{$\chi$},\Upsilon)$ & The summarized index \textit{EquiWing} %\raisebox{1.5pt}{$\chi$}, is the set of super nodes and $\Upsilon$ is the set of super edges. 
\\
\hline
$EW\raisebox{.5pt}{-}C(\raisebox{1.5pt}{$\chi$},\Upsilon)$ & The compressed index \textit{EquiWing-Comp} %\raisebox{1.5pt}{$\chi$}, is the set of super nodes divided in levels and $\Upsilon$ is the set of super edges.
\\
\hline
$u,w,u_i$($v,x,v_i$) & vertices $\in U$ (vertices $\in V$)\\
\hline
$\nu,\mu,\nu_i,\mu_i$ & super nodes $\in \raisebox{1.5pt}{$\chi$}$\\
\hline

%$\Gamma(\nu)$ & The set of neighboring vertices $\nu \in \raisebox{1.5pt}{$\chi$}$.\\
%\hline

$\psi(e)$ & The wing number of an edge $e \in E$\\
\hline
$\hour_{uvwx}$ & The butterfly composed of $u,w \in U$ and $v,x \in V$\\
\hline
%$N(e)/ \Gamma (u)$ & The set of neighboring edges/vertices of $e / u $ of a graph.\\
$ \Gamma (u)$ & The set of neighboring vertices of $u$ of a graph\\
\hline
$e_1 \equi e_2$ & $e_1$ and $e_2$ are $k$-wing equivalent\\
\hline
$e_1 \xLeftrightarrow{k} e_2$, $\hour_1 \xLeftrightarrow{k} \hour_2$& $e_1$ and $e_2$ ($\hour_1$ and $\hour_2$) are $k$-butterfly connected\\

\hline

$\nu_1 \xLeftrightarrow{\geq k} \nu_2$, $\hour_1 \xLeftrightarrow{\geq k} \hour_2$ & $\nu_1$ and $\nu_2$ ($\hour_1$ and $\hour_2$) are $k$-butterfly loose connected\\
\hline
\end{tabular}
}
\label{table:symbols}
\vspace{-5pt}
\end{table}

To speed up the personalized $k$-wing search for real-world applications which may have an extensive number of queries, we present a novel wing-based containment index. Using the proposed index, a personalized $k$-wing 
%TODO12:fixed
query can be processed in linear time w.r.t. the size of the result.  
The core idea of the index is the introduction of the \textit{$k$-butterfly equivalence} relationship among the edges of a bipartite graph. For a given bipartite graph $G$, two edges are \textit{$k$-butterfly equivalence}, if they participate in the same butterfly or are connected by %TODO13:fixed
a series of dense butterflies, i.e. \textit{$k$-butterfly connected}.
We group the edges into different \textit{containment class} using \textit{$k$-butterfly equivalence}. 
We show that %each edge in a bipartite graph can participate only in a single \textit{containment class}, and thus 
the whole bipartite graph can be partitioned into the collection of the \textit{containment classes} without losing any edges. 
% \textit{EquiWing} is the super graph index which is composed of these connected \textit{containment classes} as its super nodes. 
%TODO14:fixed
\textit{EquiWing} is the super graph based index which is composed of super nodes containing the group of $k$-butterfly equivalence edges. 
Further, we improve the space and query time of \textit{EquiWing} by constructing \textit{EquiWing-Comp}. 
%TODO15:
\textit{EquiWing-Comp} utilizes \textit{$k$-butterfly loose connectivity} %which is the relaxed version of \textit{$k$-butterfly connectivity}, 
to provide the compression (reducing number of super nodes and edges) in \textit{EquiWing}. We also integrate the hierarchy property of \textit{$k$-wing} to improve the 
%TODO16: fixed
query processing.

Since, in many real-world applications, such as online social networks \cite{kumar2010structure}, web graph \cite{oliveira2007observing} and collaboration network \cite{abbasi2012betweenness}, bipartite graphs are evolving where vertices/edges will be inserted/deleted 
%TODO18: fixed
over time dynamically. Hence, we also examine 
%TODO19: fixed
\textit{personalized $k$-wing search} in dynamic bipartite graphs. 
The insertion of an edge may affect the other edges by leading them to participate in a new \textit{$k'$-wing} cohesive subgraph ($k'\ne k$). 
Unlike the case of \textit{$k$-truss} \cite{huang2014querying} (the counter part of $k$-wing for general graphs), where an edge could move to $(k+1)$-truss subgraph after an insertion, the edge in a $k$-wing could move to ($k+i$)-wing subgraph after the insertion. Here, $i$ could be greater than $1$, which brings new challenges for $k$-wing maintenance.  
A tight upper bound needs to be derived to estimate the \textit{$(k+i)$-wing} subgraph that an edge could move to after an insertion, so that we can effectively identify the affected region with a low cost. 
Similarly, we need to derive a tight lower bound for the case of deletion. 
%TODO25: fixed
%We design efficient algorithms to update the edges in the affected region. 
%The incremental update algorithms effectively support the query for cohesive subgraph search in highly dynamic bipartite graphs.

We summarize our contributions as follows:
\begin{enumerate}
    \item \textit{Effective cohesive subgraph model }- we propose the personalized \textit{$k$-wing} model which ensures high cohesiveness. % and requires only a polynomial runtime to enumerate.  
    %a \textit{$k$-wing} based \textit{personalized bipartite cohesive subgraph search} model which ensures high cohesiveness and requires only a polynomial runtime to enumerate. %We then propose two novel offline indexing schemes for efficient query processing.
    \item \textit{EquiWing} - we present the \textit{$k$-wing equivalence} to construct the \textit{EquiWing} index. It summarizes the input bipartite graph into a super graph and is self-sufficient, which guarantees the personalized \textit{$k$-wings} can be found in linear time. 
    %the query processing in linear time. 
    %TODO:
    \item \textit{EquiWing-Comp} - we propose the \textit{$k$-butterfly loose connectivity} to compress the \textit{EquiWing} and use the hierarchy property of the \textit{$k$-wing} to expedite the query processing.
    \item \textit{Efficient maintenance of the indices} - we provide an efficient incremental algorithm to maintain \textit{EquiWing} and \textit{EquiWing-Comp} for dynamic bipartite graphs by avoiding unnecessary recomputation. 
    \item \textit{Extensive Experimental Analysis} - we carry out extensive experimental analysis of both the indexing schemes and a baseline algorithm across 4 real-world datasets. %We compare it with  \textit{index-free} and \textit{EquiWing} algorithms across all the 4 datasets.
\end{enumerate}
%TODO26: fixed
%TODO27: 
%TODO28: fixed
%TODO29: fixed
%TODO30: fixed
%TODO31: fixed
The remaining paper is organized as follows. Section \ref{section:background} discusses the preliminaries and problem definition. 
%Section \ref{section:methodology} covers the proposed methodology to tackle the problem. 
In Section \ref{section:methodology}, we propose the two indexing scheme \textit{EquiWing} and \textit{EquiWing-Comp} for personalized $k$-wing search. 
Section \ref{section:dynamic} presents an efficient algorithm to maintain \textit{EquiWing} and \textit{EquiWing-Comp}. Section \ref{section:experiments} presents the experimental evaluation of all the algorithms. Section \ref{section:related_work} states the related work. Section \ref{section:conclusion} concludes the paper.
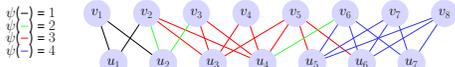
\begin{figure}[]
\centering
\scalebox{0.22}{

\begin{tikzpicture}
\node (info) at (-3,6) []{\Huge $\psi(-)=1$};
\node (info) at (-3,5.25) []{\Huge $\psi(\textcolor{green}{-} )=2$};
\node (info) at (-3,4.5) []{\Huge $\psi(\textcolor{red}{-} )=3$};
\node (info) at (-3,3.75) []{\Huge $\psi(\textcolor{blue}{-} )=4$};
\foreach \name/ \x in {v_1/1,v_2/4,v_3/7,v_4/10,v_5/13,v_6/16,v_7/19,v_8/22}
\node (\name) at ( \x,6) [circle,draw=blue!20,fill=blue!20,very thick,inner sep=7pt] {\Huge $\name$};

\foreach \name/ \x in {u_1/2,u_2/5,u_3/8,u_4/11,u_5/14,u_6/17,u_7/20}
\node (\name) at (\x,3) [circle,draw= blue!20,fill=blue!20,very thick,inner sep=7pt] {\Huge $\name$};
%\node (u3) at (5,6) [circle,draw=red!50,fill=red!20,very thick,inner sep=7pt] {$u_3$};
\path[]
(v_1) 
edge  [very thick](u_1)
edge  [very thick](u_2)
(v_2) edge  [very thick] (u_1)
edge [very thick,green] (u_2)
edge  [very thick,red] (u_3) 
edge  [very thick,red] (u_4)
(v_3) edge [very thick,red] (u_4)
edge  [very thick,green](u_2)
edge  [very thick,red](u_3)
(v_4) 
edge  [very thick,red](u_4)
edge  [very thick,red](u_3)
(v_5) edge  [very thick,red](u_3)
edge  [very thick,red](u_4)
edge  [very thick,red](u_5)
edge  [very thick,red](u_6)
(v_6) edge  [very thick,blue](u_5)
edge  [very thick,blue](u_6)
edge  [very thick,blue](u_7)
edge  [very thick,green](u_4)
(v_7) edge  [very thick,blue](u_5)
edge  [very thick,blue](u_6)
edge  [very thick,blue](u_7)
(v_8) edge  [very thick,blue](u_5)
edge  [very thick,blue](u_6)
edge  [very thick,blue](u_7);
\end{tikzpicture}%
}

\caption{Wing number for the edges in $G$}
\label{fig:my_label_ego}
\vspace{-10pt}
\end{figure}    
\section{Preliminaries and Problem Definition}
\label{section:background}
% Some of the other distinguish works for defining bipartite cohesive subgraphs are \textit{$(\alpha,\beta )$-core} \cite{ding2017efficient, liu2020efficient}, and \textit{maximal biclique} community \cite{gmati2019bi}. However, we selected \textit{$k$-wing} as the foundation structure of our bipartite cohesive subgraphs because of the following reasons: (1) \textit{$k$-wing} structure provide more cohesive and stable relationships than the naive metrics of the number of vertices/edges in the bipartite graphs \cite{DBLP:conf/wsdm/SariyuceP18}. (2) \textit{$k$-wing} follows the hierarchical modeling, wherein by changing the value of $k$, we can obtain a collection of \textit{$k$-wings} which details different levels of cohesiveness (cohesiveness of a subgraph is proportional to $k$) for the given query vertex \cite{DBLP:conf/wsdm/SariyuceP18}. (3) finding all the \textit{$k$-wings} in a bipartite graph $B$ requires a polynomial runtime \cite{DBLP:conf/wsdm/SariyuceP18}, while the denser subgraph models based on maximal biclique require exponential time \cite{kumar1999trawling, rome2005towards}.

In this section we formally introduce a $k$-wing by firstly defining a butterfly. 

\begin{definition}{Butterfly ($\hour$):}
Given a bipartite graph $G=(U,V,E)$ and four vertices $u, w\in U, v,x \in V $, a butterfly induced by $u,v,w,x$ is a cycle of length of $4$ consisting of edges $(u,v)$, $(w,v)$, $(w,x)$ and $(u,x)$ $\in$ $E$. 
%; that is, $u$ and $w$ are all connected to $v$ and $x$, respectively, via edges $(u,v)$, $(w,v)$, $(w,x)$ and $(u,x)$ $\in$ $E$. 
\label{def:butterfly}
\end{definition}
%\begin{example}
% For a given bipartite graph in Fig. \ref{fig:my_label_motivation} edges $(v_1,u_1)$, $(u_1,v_2)$, $(v_2,u_2)$ and $(u_2,v_1)$ form a butterfly $\hour_{v_1u_1v_2u_2}$ as $v_1$, $v_2$ and $u_1$, $u_2$ form a (2,2)-biclique. 
For a given bipartite graph in Fig. \ref{fig:my_label_ego} edges $(u_1,v_1)$, $(u_1,v_2)$, $(u_2,v_2)$ and $(u_2,v_1)$ form a butterfly $\hour_{u_1v_1u_2v_2}$. % as $u_1$, $u_2$ and $v_1$, $v_2$ form a (2,2)-biclique.
%\end{example}
%The butterfly has been highly recognized as a motif that represents a strong and stable relationship in a bipartite graph. Fig. \ref{fig:my_label_motivation} displays an author-paper network, we can deduce that a butterfly based model can find more cohesive and relative structures than any other motifs.
%Once the motif of a subgraph is defined we now form a bipartite cohesive subgraph using these motifs. %Eventually, \textit{$k$-bitruss} \cite{zou2016bitruss} was introduced, which groups the edges using the number of \textit{butterflies} each edge participates.

\noindent\textbf{Butterfly support $(\hour(e,G))$}. It is the number of butterflies in $G$ containing an edge $e$ and is denoted as $\hour(e,G)$. We replace $\hour(e,G)$ with $\hour(e)$ whenever the context is obvious. %$\hour(e)$ is known as the butterfly support of $e$.  

\noindent\textbf{Butterfly adjacency}. Given two butterflies $\hour_{1}$ and $\hour_{2}$ in $G$,  $\hour_{1}$ and $\hour_{2}$ are butterfly adjacent if $\hour_{1}$ $\cap$ $\hour_{2}$ $\ne$ $\emptyset$.  

\noindent\textbf{Butterfly connectivity $(\xLeftrightarrow{})$}. 
Given two butterflies $\hour_{s}$ and $\hour_{t}$ in $G$,  $\hour_{s}$ and $\hour_{t}$ are butterfly connected if there exists a series of butterflies $\hour_{1}, \ldots, \hour_{n}$  in $G$, in which $n \geq 2$ such that $\hour_{s}=\hour_{1}$, $\hour_{t}=\hour_{n}$  and for $1\le i< n$, $\hour_{i}$ and $\hour_{i+1}$ are butterfly adjacent. 

In Fig. \ref{fig:my_label_ego}, the edge $(u_2, v_2)$ is contained in $\hour_{u_1v_1u_2v_2}$, $\hour_{u_2v_2u_4v_3}$ and $\hour_{u_2v_2u_3v_3}$ thus its \textit{butterfly support} $\hour((u_2, v_2),G) = 3$. Also, $\hour_{u_1v_1u_2v_2}$ and $\hour_{u_2v_2u_3v_3}$ are said to be \textit{butterfly connected} as  $\hour_{u_1v_1u_2v_2} \cap \hour_{u_2v_2u_3v_3}=(u_2,v_2)$. %$\hour_{v_1u_1v_2u_2}$ and $\hour_{v_2u_2v_3u_3}$ are as they share a common edge $(v_2, u_2)$. 
$\hour_{u_1v_1u_2v_2}$ and $\hour_{u_3v_2u_4v_3}$ are \textit{butterfly connected} through $\hour_{u_2v_2u_3v_3}$ in $G$.

Next, we define \textit{$k$-wing bipartite cohesive subgraphs} \cite{DBLP:conf/wsdm/SariyuceP18}.

\begin{definition}{$k$-wing:}
A bipartite subgraph $H=(U,V,E) \subseteq G$ is a $k$-wing if it satisfies the following conditions: 
\begin{enumerate}
    \item Minimum butterfly constraint: $\forall$ $e$ $\in$ $E(H)$, $\hour(e,H)\ge k$. 
    \item Butterfly Connectivity: $\forall$ $e_{1}$, $e_{2} \in E(H)$, $\exists$ $\hour_{1}$ and $\hour_{2}$ $\in H$ such that $e_{1}\in \hour_{1}$, $e_{2}\in \hour_{2}$, then either $\hour_{1}=\hour_{2}$ or $~\hour_{1}$ and $~\hour_{2}$ are butterfly connected.
    \item Maximality: There is no $H'\subseteq G$ such that $H'$ satisfies the above two conditions while  $H\subseteq H'$. 
    %each edge pair $(u_1,v_1)$, $(u_2,v_2)$ $\in$ E, $(u_1,v_1)$ connects to $(u_2,v_2)$ by serises of butterflis.% by series of butterflies,
    %\item H is maximal, i.e., there is no other k-wing that subsumes H.
\end{enumerate}
\label{def:wing}
\end{definition}
%Since each edge in a bipartite graph can participate in multiple \textit{$k$-wing} cohesive subgraphs, hence a corresponding \textit{wing number} of an edge is assigned.
Further, we define the \textit{wing number} for an edge as follows.
\begin{definition}{Wing number ($\psi(e)$):}
For an edge $e$ $\in$ $E$, $\psi(e)$ is the maximum possible $k$ such that there exists a $k$-wing in $G$ containing $e$. % i.e. $\psi(e)$ $=$ $\max \{k| e \in k-wing \}$.
\label{def:wing_num}
\end{definition}

Note, given $e\in E(G)$, $\psi(e)\le \hour(e,G)$. For instance, in Fig.~\ref{fig:my_label_ego}, the wing number of $(u_2,v_2)$ is $2$. Whereas its support in the graph is $3$.

%The mathematical equivalent of Definition \ref{def:wing_num} for a \textit{wing number} can also be interpreted as $\psi(e)$ $=$ $\max \{k| e \in k-wing \}$

\begin{example}
Fig. \ref{fig:my_label_ego} shows all the edges $e$ in the bipartite graph $G$, with their respective wing number $\psi$. We observe a $4$-wing in Fig. \ref{fig:my_label_ego} (blue edges) and verify that every edge in a $4$-wing is contained in at least 4 butterflies, any two edges in a $4$-wing are reachable through adjacent
%TODO32: fixed 
butterflies, and $4$-wing is maximal as well. 
%TODO33: fixed
We can also observe that the edges $(u_4,v_5)$ and $(u_5,v_5)$ are having the same wing number of $3$. However, they belong to two different $3$-wings since they are unreachable through adjacent butterflies,  
\end{example}
%From Definition \ref{def:wing}, we can infer that the wing based bipartite cohesive subgraphs can be overlapping allowing a vertex to participate in multiple $k$-wings. 
%TODO34: fixed %TODO35: fixed
We now propose the following personalized $k$-wing search problem.
% From Definition \ref{def:wing}, we can infer that the wing based bipartite cohesive subgraphs allow a vertex to participate in multiple $k$-wings. Hence, we propose the following personalized $k$-wing search problem.
\newline
\\
\noindent\textbf{Problem Definition.} 
Given a bipartite graph $G=(U,V,E)$, a query vertex $q \in U\cup V$ and an integer $k\geq1$, we return all \textit{$k$-wings} containing $q$.
\section{Methodology}
\label{section:methodology}
In this section, we discuss the methods to address the problem of personalized $k$-wing search. 
Firstly, we describe the importance of \textit{bitruss decomposition} and the \textit{BaseLine} approach to address the problem. We also examine the limitations of the \textit{BaseLine} approach. Secondly, we propose \textit{EquiWing} index to address the limitations of \textit{BaseLine}. \textit{EquiWing} is constructed by grouping the edges in a bipartite graph based on \textit{$k$-butterfly connectivity}. It provides the summarized graph which is used for efficient query processing. Thirdly, we propose \textit{$k$-butterfly loose connectivity} and exploit the hierarchy in the $k$-wing to create \textit{EquiWing-Comp}. It further reduces the size of \textit{EquiWing} and improves the query processing.
% Firstly, we describe the importance of \textit{bitruss decomposition} and the \textit{BaseLine} approach to address the problem. Secondly, \textit{EquiWing}, a self-sufficient index that is constructed by grouping the edges in a bipartite graph based on \textit{$k$-butterfly connectivity}. It provides the summarized graph which is used for efficient query processing. Thirdly, \textit{EquiWing-Comp}, we propose \textit{$k$-butterfly loose connectivity} and exploit the \textit{hierarchy} in the \textit{$k$-wing} to create \textit{EquiWing-Comp}. It further reduces the size of \textit{EquiWing} and improves the query processing.
%we examine the limitations of the \textit{BaseLine} approach. Thirdly, we propose two indexing schemes to tackle those limitations: (i) 
\subsection{Baseline}
Observing the wing number for an edge is the same as its bitruss number, the baseline solution applies the state-of-the-art bitruss decomposition algorithm for computing the wing number for every edge offline and then performs an online search for answering a personalized $k$-wing query. \\
\noindent\textbf{Offline wing number computation}. The bitruss decomposition algorithm from \cite{DBLP:conf/icde/Wang0Q0020} is used to compute the wing number for all the edges. 
%The time complexity of the bitruss decomposition is $\mathcal{O}(|V||E|)$ for a given bipartite graph $G=(U, V, E)$ ($|V|<|U|$).
The time complexity of the bitruss decomposition is $\mathcal{O}(d^{2}_{max}|E|)$ for a given bipartite graph $G=(U, V, E)$ where $d_{max}$ is the maximum degree. 
%TODO37: 
However, we omit a detailed discussion of \cite{DBLP:conf/icde/Wang0Q0020} for the sake of brevity. 
%TODO38: fixed
Fig. \ref{fig:my_label_ego} shows $\psi(e)$, $\forall e \in E$.
%The bitruss decomposition algorithm from \cite{DBLP:conf/icde/Wang0Q0020} is used to discover all the \textit{$k$-wings} in the bipartite graph and compute their corresponding wing number for all the edges. The time complexity of the bitruss decomposition is $O(|E|^2)$ for a given bipartite graph $G=(U, V, E)$. However, our problem does not require a detailed discussion of \cite{DBLP:conf/icde/Wang0Q0020}, as a result, not included for the sake of brevity. Fig. \ref{fig:my_label_ego}, represents $\psi(e)$, $\forall e \in E$.% edges with their respective  $\psi$.
 
%We compute intersections for each vertex $r$ sharing neighbors with $q$ (say $q,r \in U$) $O(|d_{max}|)$. For each pair $q,r$ we select the vertex pair $s,t \in V$ which costs $O(|d_{vmax}|^2)$, where $d_{umax}/d_{vmax}$ is the maximum degree vertex in set $U/V$. 
\noindent\textbf{Online search}. For a query vertex $q$ and an integer $k$, the \textit{BaseLine} approach starts the search with every edge $e$ incident on $q$ satisfying $\psi(e)$ $\ge$ $k$. Such incident edges act as the initialized seed edges and %is added into a queue $Q$. A BFS traversal using $Q$ 
a BFS is performed starting with these seed-edges. %The cohesive $k$-wing subgraph 
For each edge $e$ in the seed, \textit{BaseLine} includes $e$ to the partial result and expands the seed as follows. Using the edge $e$ to be included as the partial result, the BFS considers every $e_{0}$ as a new seed-edge if $e_{0}$ satisfies the constraints: (i) $\psi(e_{0}) \geq k$; (ii) $e_{0}$ forms a butterfly with $e$ (all edges in this butterfly has wing number no less than $k$); (iii) $e_{0}$ is neither in the seed nor in the partial result. 
%The $k$-wing result are progressively generated by including $e$ during the expansion. %, where $e_0 \in k$-wing subgraph. 
%Similarly all the other $k$-wing subgraphs are enumerated using the incident edges. 
The expansion stops when no edge can be further included into the partial result, i.e., there is no seed-edge left. 
Due to the constraints applied during the search, edges satisfying the butterfly connectivity can be grouped together inherently. 
As such, the set of \textit{$k$-wings} containing $q$ can be derived by \textit{BaseLine} as the query result. 
The time complexity of \textit{BaseLine} is %$(O(|d_{umax}||d_{vmax}|^2)$ or can be written as 
$\mathcal{O}(d^{2}_{max}|E(H)|)$, where $d_{max}$ is the maximum degree and $|E(H)|$ is the number of edges in the \textit{$k$-wings} containing $q$. 
The extra cost of $\mathcal{O}(d^{2}_{max})$ compared to the general BFS is induced by checking the butterfly connectivity constraint.  
%This is because when processing an edge $e$, we have to include every edge such that it involves which takes $\mathcal{O}(d_{max})$ time. 
%no less than $k$ butterflies, it is unvisited, and it forms a butterfly with $e$, to ensure the correctness,  

\noindent \textbf{Limitations of the BaseLine approach.}
%\label{subsubsection_limitaions}
For any edge $(u,v)$ in a butterfly the algorithm needs to check all the edges incident on $u$ and $v$ and see if they satisfy all the constraints. 
%the set of three-edge set $(u,x),(w,x),$ and $(w,v)$, as well to form a $k$-wing. 
This leads to the following two unnecessary operations.
%\begin{enumerate}
    %\item 
    (i) \textit{Overhead of accessing ineligible edges:} during the search, if the wing number of an edge is less than $k$, it would not be included in the $k$-wing. Therefore an extra unnecessary overhead is required to check the ineligible edges. 
    %any of the edges wing number is less than $k$ i.e. $\psi((u,x)) < k$, $\psi((w,x)) < k$ or $\psi((w,v)) < k$ then it would not be included in the bipartite cohesive subgraph. Therefore an extra unnecessary overhead is required to check the ineligible edges. 
    %\item
    (ii) \textit{Redundant access of eligible edges:} if an edge $e$ is added into the partial result or seed,
    %TODO39: fixed
    it is accessed at least extra $3\times k$ times in the BFS, 
    %TODO40: fixed
    which is an overhead. This is because for each eligible edge $e$ ($\psi(e) \geq k$) of a butterfly, it will be accessed three times while doing the BFS from the other three eligible edges in the same butterfly. 
%\end{enumerate}
%Considering the following drawbacks, we propose a novel \textit{Equivalence Directed Tree Index}(DCTI), which overcomes the following drawbacks. It exploits the unique properties of k-wing. We classify the edges with respect to the connectivity and their wing number. To create DCTI, we first create an index \textit{EquiWing}, which is the adaptation of \textit{EquiTruss} for bipartite graphs. Further, we propose improving \textit{EquiTruss} by removing redundant nodes and integerate the structure with the hierarchy property of a \textit{k-wing}.

\subsection{Wing Equivalence}

%TODO41: fixed    
The limitations of \textit{BaseLine} can be addressed by 
%TODO42: fixed?  
an approach to be proposed, which uses one-time quick access only to the eligible edges. 
The intuitive solution is to group all the eligible edges with the same wing number and then access them. However, it is not necessary for two edges with the equal wing number to co-exists in the same \textit{$k$-wing}. Therefore, we exploit the notion of \textit{butterfly connectivity} and extend it to propose \textit{$k$-butterfly connectivity}. 
%TODO43: fixed
Further, the equivalence relationship from \cite{akbas2017truss} is adopted, and based on it, we propose \textit{$k$-wing equivalence}. The equivalence class provides the grouping of only those eligible edges which are $k$-butterfly connected, hence, can co-exist in the same \textit{$k$-wing}. % We identify a fundamental equivalence relation for edges that are connected in a \textit{$k$-wing bipartite cohesive subgraph}.  
% To address the limitations of the \textit{BaseLine} approach and enable efficient bipartite cohesive subgraph search, we adapt the notion of equivalence relationship from \cite{akbas2017truss} and propose \textit{$k$-wing equivalence}. We identify a fundamental equivalence relation for edges that are connected in a \textit{$k$-wing bipartite cohesive subgraph}. 
% for a wing, 
% As a result, an equivalence based index \textit{EquiWing}, can be developed. We use the Definition \ref{def:wing} and algorithm from \cite{DBLP:conf/icde/Wang0Q0020} for assigning the wing number ($\psi$) to all the edges in a bipartite graph (Fig. \ref{fig:my_label_ego}). Once, now we define a stronger \textit{butterfly-connectivity} constraint: \textit{k-butterfly connectivity}, as follows.
As a result, an equivalence based index \textit{EquiWing}, can be developed. To start with, we assign the wing number ($\psi$) to all the edges in a bipartite graph using the algorithm in \cite{DBLP:conf/icde/Wang0Q0020}. 
%TODO44: fixed
Now, we define a stronger \textit{butterfly-connectivity} constraint: \textit{$k$-butterfly connectivity}, as follows.
\begin{definition}
\label{def_k-butterfly}
%TODO45: fixed
{$k$-butterfly:} A butterfly $\hour_{uvwx}$ in $G$ is a
%TODO46: fixed? check all
$k$-butterfly, if $min$ $\{\psi((u,v)),$ $ \psi((u,x)), \psi((w,x)), \psi((w,v))\} \geq k$.%, then $\hour_{uvwx}$ is a k-butterfly.
\end{definition}

\begin{algorithm}[t]
\small
%\algsetup{linenosize=\small}
 % \scriptsize
	\SetKwInOut{Input}{Input}
	\SetKwInOut{Output}{Output}
	\Input{$G=(U,V,E)$} 
%	\Output{Index : \DJ $( \raisebox{1.5pt}{$\chi$},\Upsilon)$} 
    \Output{Index : $EW( \raisebox{1.5pt}{$\chi$},\Upsilon)$} 
	     \SetKwFunction{FMain}{ Index Construction  }
    \SetKwProg{Fn}{Function}{}{}
    \Fn{\FMain{G=(U,V,E)}}{
       %/*Initialization*/\\
        % \ForEach{$e\in E$}{
        % e.visited $\longleftarrow{}$ FALSE; 
        % e.list $\longleftarrow \phi$\\
        % \If{$\psi(e)=k$}{
        % $\Phi_k\longleftarrow \Phi_k \cup \{e\}$}
        % }
        $bitruss\_Decomposition(G)$;\\
       $initialize\_all\_edges()$;\\
        % %/*Index Construction*/\\
        \For{$k \longleftarrow 1$ to $k_{max}$}{
        \While{$\exists e \in \Phi_k$}{
        $e.visited=TRUE$;\\
        Create a super node $\nu$ with $\nu.snID \leftarrow ++snID$;\\
        %\tcc{We use $\raisebox{1.5pt}{$\chi$}$ for \textit{EquiWing} and replace it with $\raisebox{1.5pt}{$\chi$}$ $[k]$ for \textit{EquiWing-Comp}}
        $\raisebox{1.5pt}{$\chi$} \leftarrow \raisebox{1.5pt}{$\chi$} \cup \{\nu\}$ or \{$\raisebox{1.5pt}{$\chi$}$ $[k] \leftarrow \raisebox{1.5pt}{$\chi$}$ $ [k] \cup \{\nu\}$\};\\
        $Q.enqueue(e)$;\\
        \While{$Q\neq \emptyset$}{
        $x(u,v)\leftarrow Q.dequeue()$;
        %$H(v)\longleftarrow H(v) \cup \nu$; 
        $\nu \leftarrow \nu \cup \{x\}$;\\
        %\tcc{Let $|\Gamma(u)| \leq|\Gamma(v)|$}
        \ForEach{$id \in x.list$}{
        Create an edge ($\mu,\nu$) where $\mu$ is an existing super node with $\mu .snID=id$.\\
        % $\mu.children \longleftarrow \mu.children \cup \{\nu$\}\\
        % $\nu.parent \longleftarrow \mu$\\
        $\Upsilon \leftarrow \Upsilon \cup \{(\mu,\nu)\}$\\
        }
        \ForEach{$e' \in$ edges incident on $x$}{%N(x)$}{
        \If{$\psi(e') \geq k$ }{%and $e'\neq x$
        $ButterFly(x,e',Q)$
        }
        }
        }
       
        }
        }
        \textbf{return $EW( \raisebox{1.5pt}{$\chi$},\Upsilon)$};
}
    %TODO	
     \SetKwFunction{FMain}{ButterFly}
    \SetKwProg{Fn}{Function}{}{}
    \Fn{\FMain{$x,e',Q$}}{
         \ForEach{$e'' \in $ edges incident on $x$ and $e'$}{% N(x)\cap N(e')$}{
         %\tcc{explore the neighbors which are not visited and form a butterfly with all the edge's $\psi >=k$}
        %  \If{$\psi(x)\geq k$ and $\psi(x'')\geq k$ and $\psi(e')\geq k$ and $\psi(e'')\geq k$}{ 
        %  Process($x'',Q$); \tcc{Similarly, call procedure for $e'$ and $e''$}
         \If{$\psi(x')\geq k$ and $\psi(e'')\geq k$}{
         \tcc{$e',e'',x,x'$ form a $k$-butterfly.}
         Process($e'',Q$);\tcc{Invoke Process, for $x'$ and $e'$}
        %  Process($e',Q$);\\
        %  Process($e'',Q$);\\
         }
         %\textbf{return} true;
         }
         %\textbf{return};
         }

	\SetKwFunction{FMain}{Process}
    \SetKwProg{Fn}{Function}{}{}
    \Fn{\FMain{$e''',Q$}}{
        \eIf{$\psi(e''')=k$}{
        \If{$e'''.visited=FALSE$}{
            $e'''.visited=TRUE$;
            $Q.enqueue(e''')$;\\
            }
        }
        {
        \If{$snID \not\in e'''.list$}{
		$e'''.list \leftarrow e'''.list \cup {snID}$}
		}

%\KwRet\;
}
\caption{Index Construction (EquiWing)}
\label{Algorithm:EquiWIng}
\end{algorithm}

\begin{definition}
\label{def_k-butterfly_connected}
{$k$-butterfly connectivity ($\xLeftrightarrow{  k}$):} Given two $k$-butterflies $\hour_{x}$ and $\hour_{y}$ in G, they are $k$-butterfly connected if there exists a sequence of $n\geq 2$ $k$-butterflies: $\hour_1,...,\hour_n$ s.t. $\hour_{x}=\hour_{1}$, $\hour_{y}=\hour_{n}$ and for
%TODO47: forall vs at least one
$1\leq i<n,\hour_i \cap \hour_{i+1}\neq \emptyset$ and $\exists  e \in \hour_i \cap \hour_{i+1}$, $\psi(e) = k$.%, \forall e \in \hour_i \cap \hour_{i+1} $.
\end{definition}

\begin{example} Consider the bipartite graph $G$ in Fig. \ref{fig:my_label_ego}, and butterflies $\hour_{u_5v_6u_6v_7}$ and $\hour_{u_6v_7u_7v_8}$. They 
%TODO48:fixed
are 4-butterfly connected as they share a common edge with the wing number $4$, i.e. $\hour_{u_5v_6u_6v_7} \cap \hour_{u_6v_7u_7v_8}=\{(u_6,v_7)\}$ and $\psi((u_6,v_7))=4$. 
%TODO49:fixed
%However, $\hour_{u_3v_3u_4v_4}$ and $\hour_{u_5v_5u_6v_6}$ are not 3-butterfly connected.
 \end{example}

%TODO50:fixed
Now, we are ready to define the \textit{$k$-wing equivalence} for a pair of $e_{1},e_{2} \in E$.

%TODO51:fixed
\begin{definition}{$k$-wing equivalence:} Given any two edges $e_1, e_2 \in E$, they are $k$-wing equivalent, denoted by $e_1$ $\equi$ $e_2$, if (1) $\psi (e_1)=\psi (e_2)=k$, and (2) $\exists$ $\hour_{1}$ and $\hour_{2}$ $\in H$ such that $e_{1}\in \hour_{1}$, $e_{2}\in \hour_{2}$, then 
%either $\hour_{1}=\hour_{2}$ or $~\hour_{1} \xLeftrightarrow{k}~\hour_{2}$.
$~\hour_{1} \xLeftrightarrow{k}~\hour_{2}$.
% are k-butterfly connected. 

\label{def:k-wing_equivalence}
\end{definition}

\begin{theorem}
$k$-wing equivalence is an equivalence relationship upon $E(G)$. 
\end{theorem}
\begin{proof}
%TODO52:fixed
To prove that $e_1 \equi e_2$ is an equivalence relationship, we prove that it is reflexive, symmetric and transitive.\\
\noindent\textbf{Reflexivity.} Consider an edge $e_1 \in E$, s.t. $\psi (e_1) = k$. Utilizing Definition \ref{def:wing}, there exists at least one subgraph $G'=(U',V', E') \subseteq G$ s.t. $e_1 \in E'$, and $\forall e \in E'$, $\psi (e) \geq k$. Since there exist at least one \textit{$k$-butterfly} $\hour \subseteq G'$ s.t. $e_1 \in  \hour$. Namely, $e_1 \equi e_1$.
\\
\noindent\textbf{Symmetry.} Consider two edges $e_1, e_2 \in E$, $e_1 \equi e_2$. That is, $\psi (e_1) = \psi (e_2) = k$, and either of the following cases holds if so $e_1 \equi e_2$: (1) $e_1$ and $e_2$ are in the same \textit{$k$-butterfly}; (2) there exist two \textit{$k$-butterflies} $\hour_1$ and $\hour_2$, such that $e_1 \in \hour_1$, $e_2 \in \hour_2$, and $\hour_1 \xLeftrightarrow{k} \hour_2$. For this case, note that \textit{$k$-butterfly connectivity} is symmetric, so $\hour_2  \xLeftrightarrow{k} \hour_1$. Namely, $e_2 \equi e_1$.
\\
\noindent\textbf{Transitivity.} Consider three edges $e_1$, $e_2$, $e_3 \in E$, s.t. $e_1 \equi e_2$ and $e_2 \equi e_3$. Namely, $\psi (e_1) = \psi (e_2) = \psi (e_3) = k$, and one of the following cases holds: (1) there exist two \textit{$k$-butterflies} $\hour_1$ and $\hour_2$, s.t. $e_1$, $e_2 \in \hour_1$ and $e_2$, $e_3 \in \hour_2$. If $e_1$ and $e_3$ are located in the same \textit{$k$-butterfly}, then $e_1 \equi e_3$. 
%TODO53:
Or else, $\hour_1 \cap \hour_2 = \{e_2\}$ and $\psi (e_2) = k$, so $\hour_1 \xLeftrightarrow{k} \hour_2$. Hence, $e_1 \equi e_3$; 
%TODO54:
(2) there exist $s$ \textit{$k$-butterflies} $\hour_{x_1}, \ldots , \hour_{x_s}$ in $G$, s.t. $e_1 \in \hour_{x_1}$, $e_2 \in \hour_{x_s}$, and all the edges joining these $s$ consecutive \textit{$k$-butterflies} are with the same \textit{wing number}, $k$. Meanwhile, there exist $t$ \textit{$k$-butterflies} $\hour_{y_1}
, \ldots , \hour_{y_t}$ in $G$ s.t. $e_2 \in \hour_{y_1}$, $e_3 \in \hour_{y_t}$ and all the edges joining these $t$ \textit{$k$-butterflies} are with the same \textit{wing number}, $k$. If $\hour_{x_s} = \hour_{y_1}$, we know that $\hour_{x_1} \xLeftrightarrow{k} \hour_{y_t}$ through a series of $(s + t - 1)$ adjacent \textit{$k$-butterflies} $\hour_{x_1}, \ldots , \hour_{x_s}, \ldots , \hour_{y_t}$. Or else, we know that $\hour_{x_s} \cap \hour_{y_1} = {e_2}$ and $\psi (e_2) = k$, so $\hour_{x_1} \xLeftrightarrow{k} \hour_{y_t}$ through a series of $(s + t)$ adjacent \textit{$k$-butterflies} $\hour_{x_1}, \ldots , \hour_{x_s}, \hour_{y_1}, \ldots , \hour_{y_t}$. Hence, $e_1 \equi e_3$.
\end{proof}

Given an edge $e \in E, \psi(e)=k$, the set of edges 
%\mathbb{C}_e=$\{$x|x$ is \textit{$k$-butterfly equivalent} to $e, x\in E$\}, 
$\mathbb{C}_e=$\{$x|x \equi e, x\in E$\}, 
defines the equivalence class of $e$. 
%TODO55: %TODO56:
The set of all equivalence classes forms a mutually exclusive and collectively exhaustive partition of $E$. Each equivalence class $\mathbb{C}_e$ is composed of edges with the same wing number, $k$, that are \textit{$k$-butterfly connected}. Therefore an equivalence class $\mathbb{C}_e$ forms the basic unit for our personalized $k$-wing search.

% \begin{example}
%  Consider the bipartite graph in Fig. \ref{fig:my_label_ego}, vertex $v_5$ is contained by two 3-wings i.e. $\mathbb{W}_1=\{(v_5, u_5),(v_6,u_5), (v_7,u_5),(v_8,u_5),(v_5, u_6),(v_6,u_6), (v_7,u_6),$ $(v_8,u_6),(v_6,u_7), (v_7,u_7),(v_8,u_7)\}$ and $\mathbb{W}_2=\{(v_2,u_3),$ $(v_2,u_4),$ $ (v_3,u_3),$ $(v_3,u_4),$ $(v_4, u_3),(v_4,u_4),$ $ (v_5,u_3),$ $(v_5,u_4)\}$. The $k$-wing bipartite cohesive subgraph network $=\mathbb{W}_1+\mathbb{W}_2$.
% \end{example}
\subsection{EquiWing ($EW$)}
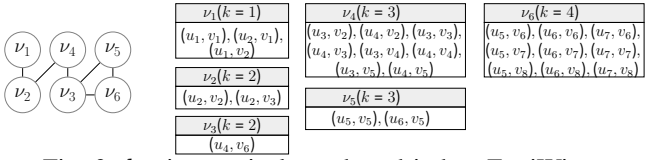
\begin{figure}[]
    \centering
\hspace{-180pt}
\begin{subfigure}{0.1\textwidth}
\scalebox{0.3}{
\centering
\begin{tikzpicture}
\node (1) at (0,2) [circle,draw=black!50,very thick,inner sep=7pt] { \Huge $\nu_1$};
\node (2) at (0,0) [circle,draw=black!50,very thick,inner sep=7pt] {\Huge $\nu_2$};
\node (3) at (2,0) [circle,draw=black!50,very thick,inner sep=7pt] {\Huge $\nu_3$};
\node (4) at (2,2) [circle,draw=black!50,very thick,inner sep=7pt] {\Huge $\nu_4$};
\node (5) at (4,2) [circle,draw=black!50,very thick,inner sep=7pt] {\Huge $\nu_5$};
\node (6) at (4,0) [circle,draw=black!50,very thick,inner sep=7pt] {\Huge $\nu_6$};
\path[]
(1) edge [very thick] (2)
(2) edge [very thick] (4)
(3) edge [very thick] (4)
edge [very thick] (5)
edge [very thick] (6)
(5)edge [very thick] (6);

\end{tikzpicture}%
}

\end{subfigure}
\hspace{.51em}
\large
\begin{subfigure}{0.\textwidth}
    \scalebox{0.5}{
    \begin{tikzpicture}
    \draw [fill=grey!20] (-2.5,2) rectangle (0.5,2.5) ;
    \node (V1) at ( -1,2.25) [] { $\nu_1(k=1) $};
    \draw [] (-2.5,1.05) rectangle (0.5,2) ;
    \node (V1') at ( -1,1.65) [] { $(u_1,v_1),(u_2,v_1),$};
    \node (V1') at ( -1,1.25) [] { $(u_1,v_2)$};
    
    \draw [fill=grey!20] (-2.5,0.3) rectangle (0.5,.8) ;
    \node (V2) at ( -1,0.55) [] { $\nu_2(k=2) $};
    \draw [] (-2.5,-.3) rectangle (0.5,0.3) ;
    \node (V2') at ( -1,0) [] { $(u_2,v_2),(u_2,v_3)$};
    
    \draw [fill=grey!20] (-2.5,-1) rectangle (0.5,-.5) ;
    \node (V3) at ( -1,-0.75) [] { $\nu_3(k=2) $};
    \draw [] (-2.5,-1.5) rectangle (0.5,-1) ;
    \node (V3') at ( -1,-1.25) [] { $(u_4,v_6)$};
    \hspace{20pt}
    
    \draw [fill=grey!20] (0.25,2) rectangle (4.5,2.5) ;
    \node (V4) at ( 2,2.25) [] { $\nu_4(k=3) $};
    \draw [] (0.25,0.55) rectangle (4.5,2) ;
    \node (V4') at ( 2.35,1.75) [] { $(u_3,v_2),(u_4,v_2),(u_3,v_3),$};
    \node (V4') at ( 2.35,1.25) [] { $(u_4,v_3),(u_3,v_4),(u_4,v_4),$};
    \node (V4') at ( 2.35,.75) [] { $(u_3,v_5),(u_4,v_5)$};

    \draw [fill=grey!20] (0.25,-.25) rectangle (4.5,0.25) ;
    \node (V5) at ( 2,0) [] { $\nu_5(k=3) $};
    \draw [] (0.25,-.85) rectangle (4.5,-.25) ;
    \node (V5') at ( 2.25,-0.55) [] { $(u_5,v_5),(u_6,v_5)$};
    
     \draw [fill=grey!20] (5,2) rectangle (9.25,2.5) ;
    \node (V5) at ( 6.75,2.25) [] { $\nu_6(k=4) $};
    \draw [] (5,0.55) rectangle (9.25,2) ;
    \node (V5') at ( 7.1,1.75) [] { $(u_5,v_6),(u_6,v_6),(u_7,v_6),$};
    \node (V5') at ( 7.1,1.25) [] { $(u_5,v_7),(u_6,v_7),(u_7,v_7),$};
    \node (V5') at ( 7.1,.75) [] { $(u_5,v_8),(u_6,v_8),(u_7,v_8)$};
 
\end{tikzpicture}%

}
\vspace{-20pt}
\end{subfigure}
\caption{$k$-wing equivalence based index, EquiWing }
          \label{fig:EquiWing}
\vspace{-10pt}    
\end{figure}
%TODO: all 'node' to 'super node'
We introduce \textit{EquiWing} based solution. We first define the \textit{EquiWing} index structure using the $k$-wing equivalence. Secondly, we devise the construction algorithm to develop the \textit{EquiWing} index. Thirdly, we propose an algorithm for the personalized $k$-wing search using \textit{EquiWing}. Last but not least, we also discuss the time and space complexities of algorithms for index construction and $k$-wing search.

\noindent\textbf{EquiWing index}. The idea behind \textit{EquiWing} is to exploit the \textit{$k$-wing} relationships for any possible $k$ wherein we form the equivalence class and summarize the given bipartite graph $G$ to a super graph $EW(\raisebox{1.5pt}{$\chi$},\Upsilon)$. 
To distinguish $EW$ from $G$, we use the term super node to represent a $\nu \in \raisebox{1.5pt}{$\chi$}$. 
Each of the super nodes $\nu \in \raisebox{1.5pt}{$\chi$}$ represents a separate equivalence class $\mathbb{C}_e$ where $e \in E$, and a super edge $(\mu, \nu) \in  \Upsilon$, where $\mu, \nu \in \raisebox{1.5pt}{$\chi$}$, represents the two equivalence classes connected that are 
%TODO56:
\textit{butterfly connected}. % \textit{$k$-butterflies}. 

\noindent \textbf{Index construction.} 
 Given a bipartite graph $G$, Algorithm \ref{Algorithm:EquiWIng} is the pseudo code to construct the \textit{EquiWing} index. We start the algorithm by computing the wing number for $\forall e \in E$ (line 2), then we initialize the values to the attributes of those edges (\textit{visited}: Boolean type, symbolizes whether the edge has been examined; \textit{list}: a set of super nodes, all the super nodes which have already been explored and are connected via series of \textit{$k$-butterflies} to the current super node) and segregate the edges in distinct sets ($\Phi_k$) based on their wing number $k$ (line 3). We then examine all the edges $e \in E$, using the sets $\Phi_k$, from $k=1$ to $k=k_{max}$ (line 4). For the selected edge $e \in \Phi_k$, we create a new super node $\nu$ for its corresponding equivalence class (line 7-8). To explore all the edges within the same equivalence class $\mathbb{C}_e$ %(edges that are \textit{$k$-butterfly connected})
 , we perform BFS using $e$ as a starting edge (line 10-17). During the exploration for the edge $e$, we also check for any previously explored $\mathbb{C}_e'$ in $e.list$, if so we form the super edges between the super nodes (line 12-14).
Meanwhile, throughout the BFS, $\forall e:\psi(e) > k$, the edge $e$ stores the current super node ($snID$) in $e.list$ (line 28-29).
%Meanwhile, throughout the BFS for all the edges $e'''$ having wing number greater than $k$ are added to the $e'''.list$  i.e. $\psi(e') \geq k$ (line 29-30).\\
\begin{example}
%TODO57:
The EquiWing of bipartite graph $G$ (Fig. \ref{fig:my_label_ego}) is shown in Fig. \ref{fig:EquiWing}. It contains 6 super nodes and each super node corresponds to a $k$-wing equivalence class in which an edge of $G$ participates. We observe in Fig. \ref{fig:EquiWing} tabulated form, that all the edges from $G$ are contained in \textit{EquiWing}, within super nodes. For example, $\nu _4$ represents $8$ edges which are $3$-butterfly connected edges in G, and form a \textit{3-wing}. Moreover, we also observe that \textit{EquiWing} contains 6 super edges depicting the butterfly connectivity between super nodes. 
\end{example}
\noindent \textbf{Time complexity.} 
Algorithm \ref{Algorithm:EquiWIng} performs BFS to enumerate all the \textit{butterflies} for each edge. Moreover, each edge needs to be checked for \textit{$k$-butterfly connectivity}, which takes $\mathcal{O}(d^{2}_{max})$. %$O(|E(G)|)$. 
Therefore, the total time complexity of Algorithm \ref{Algorithm:EquiWIng} is $\mathcal{O}(d^{2}_{max}|E(G)|)$.

\noindent\textbf{Space complexity.} For an edge $e \in E$, $|e.list|$ is the number of super nodes which are butterfly connected to $\nu$, where $e \in \nu$. Therefore, $e.list$ can take at most $\mathcal{O}(|E(G)|)$ space. This memory is free once $e$ is processed. Hence, the space complexity of Algorithm \ref{Algorithm:EquiWIng} is $\mathcal{O}(|E(G)|)$.

Note that during the creation of the \textit{EquiWing} all the \textit{$k$-wings} are completely retained, and the $\xLeftrightarrow{}$ between two \textit{$k$-wings} is also maintained through the super edges. Therefore, we consider \textit{EquiWing} as self-sufficient to identify all the \textit{$k$-wings} as it contains all the critical information.

\begin{algorithm}[t]
\small
%\algsetup{linenosize=\small}
 %\scriptsize
	\SetKwInOut{Input}{Input}
	\SetKwInOut{Output}{Output}
	\Input{$EW( \raisebox{1.5pt}{$\chi$},\Upsilon)$, $q$ and hash structure $H(q)$.} 
	\Output{$\mathbb{W}$: all $k$-wings containing $q$.} 
	\SetKwFunction{FMain}{EquiWing\_search }
    \SetKwProg{Fn}{Function}{}{}
    \Fn{\FMain{$EW( \raisebox{1.5pt}{$\chi$},\Upsilon)$,q,H(q)}}{
        % /*Initialization*/\\
        % \ForEach{$\nu \in \raisebox{1.5pt}{\chi}$}{
        % $\nu.visited \longleftarrow FALSE$
        % }
        $mark\_all\_nodes\_false()$; $l\longleftarrow 0$;\\
       \tcc{Do BFS on EquiWing using seeds. }
        \ForEach{$\nu \in H(q)$}{
        \If{$\psi(\nu) \geq k$ and $\nu .visited=FALSE$ }{
        $Q \longleftarrow \emptyset$; $Q.enqueue(\nu)$;  $l \longleftarrow l+1$;\\
        $\nu.visited=TRUE$; $\mathbb{W}_l \longleftarrow \emptyset$;\\ 
        \While{$Q \neq \emptyset$}{
        $\nu\longleftarrow Q.dequeue()$; $\mathbb{W}_l \longleftarrow \mathbb{W}_l \cup \{e|e\in \nu\}$;\\
        
        \ForEach{$(\nu,\mu) \in \Upsilon$}{
        \If{$\psi(\mu)\geq k$ and $\mu.visited=FALSE$}{
        $\mu.visited=TRUE;$ $Q.enqueue(\mu);$
        }
        }
        }
        }
        }
  \textbf{return} \{$\mathbb{W}_1$, ..., $\mathbb{W}_l$\};      
}

\caption{Cohesive group search using EquiWing}
\label{Algorithm:Community_EquiWing}
\end{algorithm}
\noindent \textbf{Query processing.} %Once \textit{EquiWing} is created, we can process our query for a bipartite cohesive subgraph search. Since \textit{EquiWing} is self-sufficient we do not need the input bipartite graph any further, thus save the time for the repeated access to the edges in G. 
Algorithm \ref{Algorithm:Community_EquiWing} represents the pseudo code for processing a personalized $k$-wing search query using \textit{EquiWing}. 
To start our search we first find the super node which contains the query vertex $q$. This is achieved by maintaining a hash structure ($H(q)$) which has the vertex as its key and the values correspond to the list of super nodes $\nu$ containing the vertex $q$. This structure can be built along with \textit{EquiWing} construction. Algorithm \ref{Algorithm:Community_EquiWing} starts with a super node $\nu \in H(q)$ with $\psi(\nu) \geq k$, then explores all the neighbors $\mu$ using BFS in \textit{EquiWing} s.t. $\forall \mu \in \Gamma(\nu), \psi(\mu) \geq k$. In the end, all the super nodes connected to $\nu$ via \textit{$k$-butterflies} form a $k$-wing stored in $\mathbb{W}_l$. The exploration continues until all $\nu \in H(q)$ are visited. 
%Algorithm \ref{Algorithm:Community_EquiWing} is a simple implementation of a BFS, hence we have worst case time complexity as $\mathcal{O}(|\Upsilon|+|\raisebox{1.5pt}{$\chi$}|)$.  
Algorithm \ref{Algorithm:Community_EquiWing} accesses only the edges in the super node $\nu$ where where $\psi(\nu) \geq k$ and each edge is accessed only once when reported as output in $\mathbb{W}_i$.
%For any edge in the super node $\mu$ where $\psi(\mu)<k$, it is not even accessed in the algorithm. 
%TODO58: fixed
Therefore, the time complexity of Algorithm \ref{Algorithm:Community_EquiWing} is $\mathcal{O}(|E(H)|)$, where $|E(H)|$ is the number of edges in the resulting $k$-wings. 
%added for comparison of EW-C runtime
%Moreover, the time complexity of Algorithm \ref{Algorithm:Community_EquiWing} can also be given as $\mathcal{O}{(|\Upsilon|+|\raisebox{1.5pt}{$\chi$}|)}$.
%\end{proof}
\begin{figure}[t]
\centering

\begin{subfigure}{0.2\textwidth}
\vspace{-1.2em}
\centering
\scalebox{0.32}{
\begin{tikzpicture}
\node (1) at (0,2) [circle,draw=black!50,very thick,inner sep=5pt] {\Huge $ \nu_1$};
\node (2) at (0,0) [circle,draw=black!50,very thick,inner sep=5pt] {\Huge$\nu_2$};
\node (3) at (2,0) [circle,draw=black!50,very thick,inner sep=5pt] {\Huge $\nu_3$};
\node (4) at (2,2) [circle,draw=black!50,very thick,fill=red!20,inner sep=5pt] {\Huge $\nu_4$};
\node (5) at (4,2) [circle,draw=black!50,very thick,fill=blue!20,,inner sep=5pt] { \Huge $\nu_5$};
\node (6) at (4,0) [circle,draw=black!50,very thick,fill=blue!20,inner sep=5pt] {\Huge $\nu_6$};
\path[]
(1) edge [very thick] (2)
(2) edge [very thick] (4)
(3) edge [very thick] (4)
edge [very thick] (5)
edge [very thick] (6)
(5)edge [very thick] (6);

\end{tikzpicture}%
}
\caption{EquiWing }
\label{fig:EquiWing_community_a}
\end{subfigure}
\hspace{1em}
\vspace{-1.7em}
\begin{subfigure}{0.22\textwidth}
\centering
\scalebox{0.21}{
\begin{tikzpicture}
\node (v_5) at ( 0,6) [circle,draw=black!50,fill=grey!20,very thick,inner sep=7pt] {\Huge $v_5$};
\foreach \name/ \x in {v_2/-6,v_3/-4,v_4/-2,v_6/2,v_7/4,v_8/6}
\node (\name) at ( \x,6) [circle,draw=black!50,very thick,inner sep=7pt] {\Huge$\name$};
\foreach \name/ \x in {u_3/-4.5,u_4/-2,u_5/.5,u_6/3,u_7/5.5}
\node (\name) at (\x,2) [circle,draw=black!50,very thick,inner sep=7pt] {\Huge $\name$};
\path[]
(v_2)
edge  [very thick,red](u_3)
edge  [very thick,red](u_4)
(v_3)
edge  [very thick,red](u_3)
edge  [very thick,red](u_4)
(v_4)
edge  [very thick,red](u_3)
edge  [very thick,red](u_4)
(v_5)
edge  [very thick,red](u_3)
edge  [very thick,red](u_4)
edge  [very thick,blue](u_6)
edge  [very thick,blue](u_5)
(v_6) edge  [very thick,blue](u_7)
edge  [very thick,blue](u_5)
edge  [very thick,blue](u_6)
(v_7) edge  [very thick,blue](u_7)
edge  [very thick,blue](u_5)
edge  [very thick,blue](u_6)
(v_8) edge  [very thick,blue](u_7)
edge  [very thick,blue](u_5)
edge  [very thick,blue](u_6);
\end{tikzpicture}%
}

\caption{$3$-wings for $v_5$}
\label{fig:EquiWing_community_b}
\vspace{15pt}
\end{subfigure}
\caption{$3$-wing search for $v_5$ using EquiWing}
%The two $3$-wing bipartite cohesive subgraphs for the query vertex $v_5$: first, $\mathbb{W}_1$ (red); second, $\mathbb{W}_2$ (blue)}
\label{fig:EquiWing_community}
\vspace{-10pt}
\end{figure}
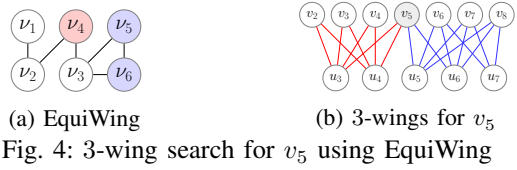
\begin{example}
Fig. \ref{fig:EquiWing_community} shows the $k$-wing search for the query vertex $q=v_5$ and $k=3$. We first locate the super nodes $H(v_5)$ i.e. $\nu_4$ and $\nu_5$. Starting with $\nu_4$, we firstly verify $\psi(\nu_4) \geq 3$. Since it satisfies the condition, we add all the edges of $\nu_4$ to $\mathbb{W}_1$, then we explore the neighbors of $\nu_4$. We find that none of $\mu \in \Gamma(\nu_4)$ has $\psi(\mu) \geq 3$, therefore, we report $\mathbb{W}_1$ as the first $k$-wing and move to $\nu_5$. Similarly, we explore $\nu_5$ and $\Gamma(\nu_5)$ and report $\mathbb{W}_2$ as the second $k$-wing. Hence, the query result is $\mathbb{W}=\{\mathbb{W}_1+\mathbb{W}_2\}$.%for the next bipartite cohesive subgraph. 
%Since $\psi(\nu_5)\geq3$, we add its edge to $\mathbb{W}_2$. Moreover, $\nu_6 \in \Gamma(\nu_5)$ has $\psi(\nu_6) \geq 3$, therefore, we add all its edges to $\mathbb{W}_2$. We report $\mathbb{W}_2$ as the second bipartite cohesive subgraph as the BFS is completed.
\end{example}

\subsection{Compressing EquiWing}
\textit{EquiWing} is an efficient index, however, there are still some limitations. Firstly, unnecessary segregation exists for super nodes which always occur together in a \textit{$k$-wing}. 
%TODO59: fixed
Secondly, the unexplored hierarchical properties of \textit{$k$-wing} can be used to reduce the index size. %Secondly, unexplored properties of $k$-wing which can be exploited to reduce the index size further. 
%Firstly, unnecessary segregation super nodes are presented i.e. for two nodes with the same value of $\psi =k$, which always occur together in a bipartite cohesive subgraph (\textit{$k$-wing}) can be combined in one, resulting in a compact index. Secondly, the overhead space for the auxiliary hash structure $H$. 
Therefore, to address these two challenges we propose the notion of \textit{$k$-butterfly loose connectivity} and exploit the hierarchical property among $k$-wings to form a more compressed index.\\
\noindent\textbf{Combining super nodes.} The notion of \textit{$k$-butterfly connectivity} is used for summarizing a bipartite graph to form \textit{EquiWing}, however, it is too strict which results in unnecessary segregation of super nodes. For example, in Fig. \ref{fig:EquiWing}, 
%TODO60:
$\nu_2$ and $\nu_3$ will participate in the same $k$-wing, thus can be combined. As a result, we propose a relaxed version of \textit{$k$-butterfly connectivity} i.e. \textit{$k$-butterfly loose connectivity}, as follows. 
%TODO: 
\begin{definition}
{$k$-butterfly loose connectivity $( \xLeftrightarrow{\geq k})$:} Two super nodes $\nu_x$ and $\nu_y$ ($\psi(\nu_x)=\psi(\nu_y)=k$) in \textit{EquiWing} are $k$-butterfly loose connected if there exists a sequence of $n > 2$ connected super nodes $\nu_1,...,\nu_n$ s.t. $\nu_x=\nu_1$, $\nu_y=\nu_n$ and 
%TODO61: fixed
$\forall \nu_i,\psi(\nu_i) \geq k$.
\label{def:loose_connectivity}
\end{definition}

%TODO62: fixed
We can observe in Fig. \ref{fig:EquiWing}, super node $\nu_2 \xLeftrightarrow{\geq k} \nu_3$. $\nu_2$ is connected to $\nu_3$ via $\nu_4$ and $\psi(\nu_4)\geq \psi(\nu_2)=\psi(\nu_3)=3$. Since we already have \textit{EquiWing} at our disposal, we use $ \xLeftrightarrow{\geq k}$ to compress the \textit{EquiWing} using the following theorem:
\begin{theorem}
\label{theorem:super_nodes}
%TODO:add \psi condition. Covered in definition 
If any two super nodes $\nu, \mu \in \raisebox{1.5pt}{$\chi$}$ are \textit{$k$-butterfly loose connected}. Then $\nu$ and $\mu$ can be combined together.
\end{theorem}
\begin{proof}

%TODO62: fixed
%TODO63: fixed
%TODO63: fixed
Let there be two super nodes $\nu, \mu \in \raisebox{1.5pt}{$\chi$}$, with $\psi (\nu)= \psi ( \mu )=k$. Using Definition \ref{def:wing}, a $k$-wing for $\nu$ will always include all connected series of super nodes $\nu_i$, with $\psi(\nu_i) \geq k$. Consequently, it will also include super nodes with $\psi(\nu_i)=k$ i.e. $\mu$, connected via series of super nodes $\nu_1,\nu_2,...,\nu_i$, where $\forall \psi(\nu_i) \geq k$. Using Definition \ref{def:loose_connectivity}, we deduce that $\nu$ and $\mu$ are $k$-butterfly loose connected. %$\nu \xLeftrightarrow{\geq k} \mu$. The same can be said for $\mu$
We can also observe that $\nu$ and $\mu$ will always occur together in the resulting $k$-wing. %Further, we know that each \textit{($k+1)$-wing} is contained by a \textit{$k$-wing} \cite{zou2016bitruss}, hence $\nu$ and $\mu$ will always co-exist in the same \textit{$(k-1)$-wing}.
Therefore, we conclude that combining them will make the index compact without compromising its correctness.
\end{proof}

\noindent\textbf{Integrating the hierarchy structure to EquiWing.} The combining of super nodes would eventually reduce the size of \textit{EquiWing}, however, comparing each pair of super nodes in \textit{EquiWing} could be costly. Fortunately, $k$-wing structure possesses a crucial hierarchy property that allows to expedite the combining of super nodes. The hierarchy property states that a \textit{($k+1$)-wing} is a subgraph of \textit{$k$-wing} \cite{zou2016bitruss}. 
% a \textit{$k$-wing} possess the following hierarchy property:
% \begin{property}
% For a given bipartite graph (k+1)-wing is a subgraph of $k$-wing \cite{zou2016bitruss}.
% \label{prop:k-wing_hierarchy}
% \end{property}
The super nodes in the index \textit{EquiWing} can be segregated into different hierarchy levels based on the $\psi$. Fig. \ref{fig:EquiWing-Comp_index} represents the levels in the \textit{EquiWing}. 
%TODO64: fixed
Since only the super nodes with same value of $\psi$ can be combined, we compare the pairs only at the same hierarchy level. Also the application of Theorem \ref{theorem:super_nodes} can be confined to each hierarchical level with at least two super nodes. The hierarchy structure allows to combine the super nodes starting from the level $k_{max}$ to the lowest level.

\subsection{EquiWing-Comp ($EW \raisebox{0.5pt}{-} C$)}
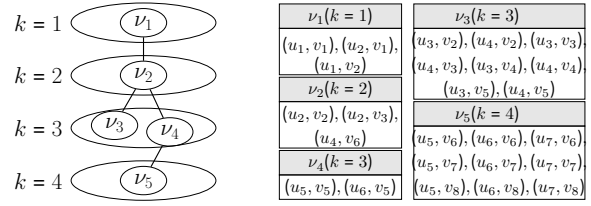
\begin{figure}[]
\begin{subfigure}{0.15\textwidth}
\centering
%\includegraphics[scale=0.3]{profile5.pdf}
%\hspace{-1em}
\scalebox{0.35}{
\begin{tikzpicture}
\node (v_1') at ( -4,6)  {\Huge $k=1$};
\draw (0,6) ellipse (2.75cm and .75cm);
\node (v_1) at ( 0,6) [ thick, draw, ellipse, minimum width=50pt,
    align=center] {\Huge $\nu_1$};
%\draw (0,6) ellipse (1cm and .5cm);

\node (v_2') at ( -4,4)  {\Huge $k=2$};
\draw (0,4) ellipse (2.75cm and .75cm);
\node (v_2) at ( 0,4) [ thick,draw, ellipse, minimum width=50pt,
    align=center]{\Huge $\nu_2$};
%\draw (0,4) ellipse (1cm and .5cm);

\node (v_5) at ( -4,0)  {\Huge $k=4$};
\draw (0,0) ellipse (2.75cm and .75cm);
\node (v_5) at ( 0,0) [ thick, draw,ellipse, minimum width=50pt,
    align=center]{\Huge $\nu_5$};
%\draw (0,0) ellipse (1cm and .5cm);

\node (v_3') at ( -4,2)  {\Huge $k=3$};
\draw (0,2) ellipse (2.75cm and .75cm);
\node (v_4) at ( 1,1.85) [ thick, draw, ellipse, minimum width=50pt,
    align=center]{\Huge $\nu_4$};
\node (v_3) at ( -1.1,2.1) [ thick, draw, ellipse, minimum width=50pt,
    align=center]{\Huge $\nu_3$};
%\draw (1,1.85) ellipse (1cm and .5cm);
%\draw (-1.1,2.1) ellipse (1cm and .5cm);
\path[]
(v_1) edge [very thick] (v_2)
(v_2) edge [very thick] (v_3)
    edge [very thick] (v_4)
   % edge [very thick] (v_5)
    (v_4) edge [very thick] (v_5);
\end{tikzpicture}
}
%\caption{EquiWing-Comp }
\label{fig:community_result_a}
%\caption{k-wing edges }
\end{subfigure}
\hspace{1.5em}
\begin{subfigure}{0.2\textwidth}
    
    \scalebox{0.65}{
    \begin{tikzpicture}
    \draw [fill=gray!20] (-2.25,2) rectangle (0.25,2.5) ;
    \node (V1) at ( -1,2.25) [] {$\nu_1(k=1) $};
    \draw [] (-2.25,1.05) rectangle (0.25,2) ;
    \node (V1') at ( -1,1.65) [] {$(u_1,v_1),(u_2,v_1),$};
    \node (V1') at ( -1,1.25) [] {$(u_1,v_2)$};
    
    \draw [fill=gray!20] (-2.25,0.5) rectangle (0.25,1) ;
    \node (V2) at ( -1,0.75) [] {$\nu_2(k=2) $};
    \draw [] (-2.25,-.45) rectangle (0.25,0.5) ;
    \node (V2') at ( -1,0.25) [] {$(u_2,v_2),(u_2,v_3),$};
    \node (V2') at ( -1,-0.25) [] {$(u_4,v_6)$};
    
    \draw [fill=gray!20] (-2.25,-1) rectangle (0.25,-.5) ;
    \node (V3) at ( -1,-0.75) [] {$\nu_4(k=3) $};
    \draw [] (-2.25,-1.5) rectangle (0.25,-1) ;
    \node (V3') at ( -1,-1.25) [] {$(u_5,v_5),(u_6,v_5)$};
    
    \draw [fill=gray!20] (0.5,2) rectangle (4,2.5) ;
    \node (V4) at ( 2,2.25) [] {$\nu_3(k=3) $};
    \draw [] (0.5,0.55) rectangle (4,2) ;
    \node (V4') at ( 2.25,1.75) [] {$(u_3,v_2),(u_4,v_2),(u_3,v_3),$};
    \node (V4') at ( 2.25,1.25) [] {$(u_4,v_3),(u_3,v_4),(u_4,v_4),$};
    \node (V4') at ( 2.25,.75) [] {$(u_3,v_5),(u_4,v_5)$};
    
    \draw [fill=gray!20] (0.5,0) rectangle (4,.5) ;
    \node (V5) at ( 2,.25) [] {$\nu_5(k=4) $};
    \draw [] (0.5,-1.5) rectangle (4,0) ;
    \node (V5') at ( 2.25,-0.25) [] {$(u_5,v_6),(u_6,v_6),(u_7,v_6),$};
    \node (V5') at ( 2.25,-0.75) [] {$(u_5,v_7),(u_6,v_7),(u_7,v_7),$};
    \node (V5') at ( 2.25,-1.25) [] {$(u_5,v_8),(u_6,v_8),(u_7,v_8)$};

\end{tikzpicture}%

}
%\caption{Nodes containing edges}
\label{subfig:tables}

\end{subfigure}
\vspace{-10pt}
\caption{Compact EquiWing index, EquiWing-Comp}
\label{fig:EquiWing-Comp_index}
\vspace{-10pt}
\end{figure}
In this section, we propose the compressed version of the \textit{EquiWing} index scheme, \textit{EquiWing-Comp}. We incorporate the \textit{$k$-butterfly loose connectivity} for combining the super nodes and exploit the hierarchy structure for finding the seeds quickly to expedite the query processing without any space overhead. We now present the index construction and query processing for \textit{EquiWing-Comp} with their respective complexities.
\\
\noindent\textbf{Index construction.}
We integrate the \textit{EquiWing} construction Algorithm \ref{Algorithm:EquiWIng}, with a hierarchy property by changing $\raisebox{1.5pt}{$\chi$}$ to  $\raisebox{1.5pt}{$\chi$}$ $[k]$ (line 13). The number of super nodes and edges in the resulting \textit{EquiWing} remain unaffected from this change. Once \textit{EquiWing} is constructed, we exploit Theorem \ref{theorem:super_nodes} to combine the super nodes with the same $\psi$ value when necessary in \textit{EquiWing}. 
We iterate across all the levels of \textit{EquiWing} from $k_{max}$ to $1$, to combine super nodes with the same $\psi$ value when necessary. Combining of the two super nodes $\nu$ and $\mu $ is performed on the created \textit{EquiWing} and is achieved by: first, move all the edges contained within the super node $\mu$ to $\nu$; second, we delete the super node $\mu$ and assign its edges to $\nu$.
% After selecting a node $\nu$ we find all the \textit{k-butterfly connected} super nodes by using function \textit{find$\_k$-wing} in line 6. The super nodes in a \textit{$k$-wing} with $\psi=k$ are of interest as they are analogous to finding the \textit{$k$-butterfly loose connected} super nodes. Searching for the \textit{$k$-wing} is simple in \textit{EquiWing}, as we utilize the index, to perform a BFS (starting with node $\nu$) for exploring the neighbors with $\psi \geq k$. \textit{find$\_k$-wing} returns the \textit{list} of nodes to be combined with $\nu$. %Line 7-10, combines all the nodes $\nu_j \in list$ to $\nu$. 
% The combining of two nodes $\nu$ and $\nu_j \in list$ is achieved by (line 7-10): first, move all the edges contained within the super node $\nu_j$ to $\nu$ (line 8); second, we delete the node $\nu_j$ and assign its edges to $\nu$ (line 10). 
% \begin{theorem}
% The construction time for EquiWing-Comp is of the order $O(|E|)$ and the space require is $O(|E|)$.
% \end{theorem}
% \begin{proof}
\\
\noindent \textbf{Time complexity.} The cost for creating \textit{EquiWing-Comp} is divided into two parts: (i) creating \textit{EquiWng} that is $\mathcal{O}(d^{2}_{max}|E(G)|)$. (ii) combining the super nodes which performs BFS over \textit{EquiWing}, and takes negligible time. Therefore, the total construction time for \textit{EquiWing-Comp} is $\mathcal{O}(d^{2}_{max}|E(G)|)$. The construction of \textit{EquiWing-Comp} requires an extra space for 
%TODO65: fixed
combining super nodes, which is not the dominating part. Hence the effective space complexity of \textit{EquiWing-Comp} construction is $\mathcal{O}(|E(G)|)$.
%The space required is less than \textit{EquiWing} due to the combining of super nodes (observed in experiments). However, in the worst case scenario, it is the same as that of \textit{EquiWing} ($O(|E|)$).
\begin{example}
%TODO66: fixed
The EquiWing-Comp for the bipartite graph $G$ is shown in Fig. \ref{fig:EquiWing-Comp_index}. We observe that all the edges from \textit{EquiWing} are contained in \textit{EquiWing-Comp}. In Fig. \ref{fig:EquiWing}, $\nu_2 \xLeftrightarrow{\geq k} \nu_3$, hence, super node $\nu_3$ is combined to $\nu_2$. In Fig. \ref{fig:EquiWing-Comp_index}, $\nu_2$ contains all the edges from $\nu_2$ and $\nu_3$ of EquiWing, which can be seen in the tabulated form in Fig. \ref{fig:EquiWing-Comp_index}. 
%TODO67: fixed
%Moreover, the number of super nodes and edges have been reduced to $5$ and $5$ in \textit{EquiWing-Comp} (Fig. \ref{fig:EquiWing-Comp_index}), from $6$ and $6$ in \textit{EquiWing} (Fig. \ref{fig:EquiWing}). 
%The super nodes and super edges in the \textit{EquiWing-Comp} still represent the same entities like that of \textit{EquiWing.}% the equivalence class and butterfly connectivity between super nodes respectively.
\end{example}
% \begin{algorithm}[]
% %\algsetup{linenosize=\small}
%  % \scriptsize
%  \small
% 	\SetKwInOut{Input}{Input}
% 	\SetKwInOut{Output}{Output}
% 	\Input{EW-C $( \raisebox{1.5pt}{$\chi$},\Upsilon)$ and the query vertex q} 
% 	\Output{$\mathbb{W}$: all $k$-wing cohesive groups} 
% 	\SetKwFunction{FMain}{Cohesive\_group\_search }
%     \SetKwProg{Fn}{Function}{}{}
%     \Fn{\FMain{EW-C$(\raisebox{1.5pt}{$\chi$},\Upsilon)$,q}}{
%         $seed(q)\longleftarrow \emptyset$\\
%         %\tcc{find all the seeds for $q$.}% i.e. all $\nu \in \raisebox{1.5pt}{$\chi$}$, where $\nu $ contains $q$ as one of its edge vertices and $\psi(\nu)>=k$.}
%         \For{$k' \longleftarrow k$ to $k_{max}$}{
%         \While{$\exists \nu \in \raisebox{1.5pt}{$\chi$}_{k'}$}{
       
%         \If{$q \in \nu$ }{
%         $seed(q) \longleftarrow seed(q) \cup \{\nu\}$
%         }
%         }
%         }
%         %$EquiWing\_search(EW-C,q,seed(q))$\\
%   \textbf{return} \{EquiWing\_search(EW-C$( \raisebox{1.5pt}{$\chi$},\Upsilon)$,q,seed(q))\}    
% }

% \caption{Cohesive group search using EquiComp}
% \label{Algorithm:Community_EquiWing_2}
% %
% \end{algorithm}
\noindent \textbf{Query processing.} 
%TODO68: fixed
%TODO69: fixed
For the given query vertex $q$, the \textit{EquiWing-Comp} index created can be used to process the personalized $k$-wing search query using $H(q)$. The index \textit{EquiWing} ($EW$) can be replaced by \textit{EquiWing-Comp} ($EW\raisebox{0.5pt}{-}{C}$ ) in Algorithm \ref{Algorithm:Community_EquiWing} to perform the $k$-wing search represented in Figure \ref{fig:EquiWing_Comp}. Moreover, the hierarchical structure of \textit{EquiWing-Comp} also allows an alternative to exploit the input $k$ for quickly finding $H(q)$ without using $H$. However, we do not discuss the alternative for concision.%the sake of brevity.
% The pseudo code for a bipartite cohesive subgraph search is shown in Algorithm \ref{Algorithm:Community_EquiWing_2}. We first discover all the \textit{seed(q)} for the query vertex \textit{q} (line 3-6). Since we know that a vertex in \textit{$k$-wing} can occur in \textit{$(k-1)$-wing}. Therefore, to find the \textit{seed(q)} we traverse across the different levels of \textit{EquiWing-Comp}, starting from $k$ to $k_{max}$. Once the \textit{seed(q)} is found the searching process is similar to that of \textit{EquiWing}, as a result in the interest of concision, we use Algorithm \ref{Algorithm:Community_EquiWing} (line 7 in Algorithm \ref{Algorithm:Community_EquiWing_2}) for completing the search. %Fig. \ref{fig:EquiWing_community} depicts the process of searching using the \textit{EquiWing-Comp} index.
% \begin{theorem}
% The runtime complexity of Algorithm \ref{Algorithm:Community_EquiWing} depends on the number of super nodes and edges in EquiWing i.e.$ O(|\Upsilon|+|\raisebox{1.5}{$\chi$}|)$.
% \end{theorem}
% \begin{proof}
\\
\noindent\textbf{Complexity}. Since we are only replacing index in Algorithm \ref{Algorithm:Community_EquiWing}, the worst case time complexity remains $\mathcal{O}{|E(H)|}$. However, the time complexity of \textit{EquiWing-Comp} reduces to $\mathcal{O}(|\raisebox{1.5pt}{$\chi$}|)$ from $\mathcal{O}{(|\Upsilon|+|\raisebox{1.5pt}{$\chi$}|)}$ in case of \textit{EquiWing}, since \textit{EquiWing-Comp} is a tree structure. 
%TODO70: fixed
Besides, it is crucial to note the number of super nodes in \textit{EquiWing-Comp} is less than \textit{EquiWing} (detailed in Section \ref{subsection:index_constructiion}).%Algorithm \ref{Algorithm:Community_EquiWing_2} can be separated into two parts. The first part for constructing the \textit{seed} (line 3-6), that is of order $O(|\raisebox{1.5pt}{$\chi$}|)$. The second part is the Algorithm \ref{Algorithm:Community_EquiWing} with the worst case time complexity as $O(|\Upsilon|+|\raisebox{1.5pt}{$\chi$}|)$. Therefore, the effective runtime for \textit{EquiWing-Comp} is $O(|\Upsilon|+|\raisebox{1.5pt}{$\chi$}|)$. However, it is crucial to note the number of super nodes in \textit{EquiWing-Comp} are very much less than \textit{EquiWing} (detailed in Section \ref{subsection:index_constructiion}).  %However, eventually we would be extracting the edges from each super node which takes constant. It  
% \end{proof}
\begin{example}
Fig. \ref{fig:EquiWing_Comp} shows the $k$-wing search for the query vertex $q=v_5$ and $k=3$. We perform the search using Algorithm \ref{Algorithm:Community_EquiWing}. The result obtain from EquiWing-Comp (Fig. \ref{fig:EquiWing_Comp_b}) is same as of EquiWing (Fig. \ref{fig:EquiWing_community_b}).
% Fig. \ref{fig:EquiWing_Comp} shows the bipartite cohesive subgraph search for the query vertex $q=v_5$ and $k=3$. We first locate the \textit{seed($v_5$)} i.e. $ \nu_3 $ and $\nu_4$, by directly accessing the $\raisebox{1.5pt}{$\chi$}[3]$ in EquiWing-Comp. Once the $seed(v_5)$ is discovered, the search becomes the same as that of \textit{EquiWing}. The result obtain from EquiWing-Comp (Fig. \ref{fig:EquiWing_Comp_b}) is same as of EquiWing (Fig. \ref{fig:EquiWing_community_b}).
\end{example}
\begin{figure}[]
\begin{subfigure}{0.2\textwidth}
\centering
%\includegraphics[scale=0.3]{profile5.pdf}
%\hspace{-1em}
\scalebox{0.26}{
\begin{tikzpicture}
\node (v_1') at ( -4,6)  {\Huge $k=1$};
\draw (0,6) ellipse (2.75cm and .75cm);
\node (v_1) at ( 0,6) [ thick, draw, ellipse, minimum width=50pt,
    align=center] {\Huge $\nu_1$};
%\draw (0,6) ellipse (1cm and .5cm);

\node (v_2') at ( -4,4)  {\Huge $k=2$};
\draw (0,4) ellipse (2.75cm and .75cm);
\node (v_2) at ( 0,4) [ thick,draw, ellipse, minimum width=50pt,
    align=center]{\Huge $\nu_2$};
%\draw (0,4) ellipse (1cm and .5cm);

\node (v_5) at ( -4,0)  {\Huge $k=4$};
\draw (0,0) ellipse (2.75cm and .75cm);
\node (v_5) at ( 0,0) [ thick, draw, fill=blue!20,ellipse, minimum width=50pt,
    align=center]{\Huge $\nu_5$};
%\draw (0,0) ellipse (1cm and .5cm);

\node (v_3'') at ( -5.3,2){\Huge $\Rightarrow$};
\node (v_3') at ( -3.9,2)  {\Huge $k=3$};
\draw (0,2) ellipse (2.75cm and .75cm);
\node (v_4) at ( 1,1.85) [ thick, draw,fill=blue!20, ellipse, minimum width=50pt,
    align=center]{\Huge $\nu_4$};
\node (v_3) at ( -1.1,2.1) [ thick, draw,fill=red!20, ellipse, minimum width=50pt,
    align=center]{\Huge $\nu_3$};
%\draw (1,1.85) ellipse (1cm and .5cm);
%\draw (-1.1,2.1) ellipse (1cm and .5cm);
\path[]
(v_1) edge [very thick] (v_2)
(v_2) edge [very thick] (v_3)
    edge [very thick] (v_4)
    %edge [very thick] (v_5)
    (v_4) edge [very thick] (v_5);
    
\end{tikzpicture}
}
\caption{$k$-wing search}
\label{EquiWing_Comp_a}
%\caption{$k$-wing edges }

\end{subfigure}
\hspace{1em}
\begin{subfigure}{0.25\textwidth}
\centering
\scalebox{0.34}{
\begin{tikzpicture}
\node (v_5) at ( 0,6) [circle,draw=black!50,fill=gray!20,very thick,inner sep=7pt] {\Huge $v_5$};
\foreach \name/ \x in {v_2/-6,v_3/-4,v_4/-2,v_5/0,v_6/2,v_7/4,v_8/6}
\node (\name) at ( \x,6) [circle,draw=black!50,very thick,inner sep=7pt] {\Huge $\name$};
\foreach \name/ \x in {u_3/-4.5,u_4/-2,u_5/.5,u_6/3,u_7/5.5}
\node (\name) at (\x,2) [circle,draw=black!50,very thick,inner sep=7pt] { \Huge $\name$};
%\node (u3) at (5,6) [circle,draw=red!50,fill=red!20,very thick,inner sep=7pt] {$u_3$};
\path[]
(v_2)
edge  [very thick,red](u_3)
edge  [very thick,red](u_4)
(v_3)
edge  [very thick,red](u_3)
edge  [very thick,red](u_4)
(v_4)
edge  [very thick,red](u_3)
edge  [very thick,red](u_4)
(v_5)
edge  [very thick,red](u_3)
edge  [very thick,red](u_4)
edge  [very thick,blue](u_6)
edge  [very thick,blue](u_5)
(v_6) edge  [very thick,blue](u_7)
edge  [very thick,blue](u_5)
edge  [very thick,blue](u_6)
(v_7) edge  [very thick,blue](u_7)
edge  [very thick,blue](u_5)
edge  [very thick,blue](u_6)
(v_8) edge  [very thick,blue](u_7)
edge  [very thick,blue](u_5)
edge  [very thick,blue](u_6);

\end{tikzpicture}%
}
\caption{$3$-wings for vertex $v_5$}
\label{fig:EquiWing_Comp_b}
\end{subfigure}
\caption{The two $3$-wings for the query vertex $v_5$}%: first, $\mathbb{W}_1$ (red); second, $\mathbb{W}_2$ (blue)}
\label{fig:EquiWing_Comp}
\vspace{-10pt}
\end{figure}
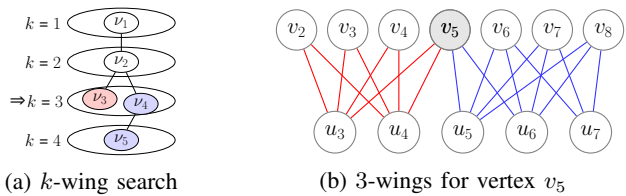
\section{Dynamic Maintenance of Indices}
\label{section:dynamic}
%When graphs are dynamically
%updated, a straightforward solution to maintain the proposed indices is reconstructing it from scratch, which is inefficient for large graphs. In this section, we discuss the incremental algorithms for maintaining the indices on dynamic graphs. We mainly focus on edge insertion and deletion, because node insertion/deletion can be reduced to as a sequence of edge insertions/deletions update problems.

In this section, we examine the $k$-wing cohesive subgraph search in a dynamic graph using the proposed indices. The dynamic change of a graph refers to the insertion/deletion of edges and vertices. We primarily focus on the edge insertion/deletion since vertex insertion/deletion can be regarded as a set of edge insertions/deletions incident to the vertex that is inserted/deleted.

%Let us consider the insertion of an edge $e'(u',v')$ in $G$. This could result in forming new butterflies $\{\hour_{ee'xx'}$ | $e,x,x' \in E\}$. Due to a new butterfly $\hour_{ee'xx'}$, the butterfly support of the other three edges $e,x$ and $x'$ is increased and this change is denoted as $\delta (\hour (e))$. The value of $\delta (\hour (e))$ can be $ \geq 1$ (Inserting $(v_5,u_6)$ in Fig. \ref{fig:my_label_ego}), explained later in Section \ref{subsection:scope_of_affected_edges}. This could further lead to increase the \textit{wing number $\psi$} of those edges. 

%new butterfly->increase of butterfly support->increase of subgraph min butterfly support. 
Let us consider the insertion of an edge $e'(u',v')$ in $G$. This could result in forming new butterflies. 
%TODO71: fixed
Due to a newly formed butterfly $\hour_{u'v'uv}$, the butterfly support for $(u,v)$, $(u,v')$ and $(u',v)$ in $\hour_{u'v'uv}$ increases. 
The increase in butterfly support may lead to the increase in the wing number for those edges. 
Moreover, the affected edges may not necessarily limit to the edges that are incident on $u'$ and $v'$. 
The edge deletion has a similar effect. 
%TODO72: fixed
Calculating the wing number from scratch is costly. 
Therefore, to handle the wing number update efficiently, the key is to identify the scope of the affected edges in the graph precisely. 
%TODO73: fixed
%TODO74: fixed
Moreover, due to the nature of the bipartite graph, when forming/breaking a butterfly $\hour_{u'v'uv}$ because of inserting/deleting an edge $(u',v')$, the butterfly support for edge $(u,v)$ in $\hour_{u'v'uv}$ increases/decreases by $1$ while the butterfly support for $(u,v')$, $(u',v)$ in $\hour_{u'v'uv}$ could increase/decrease more than $1$, which brings new challenges and makes the well studied theoretical results for core and truss maintenance inapplicable to wing maintenance. We conduct a novel study on the wing maintenance problem.

The increase in wing number due to the insertion of an edge requires to update the indices %$EquiWing$ and $EquiWing-Comp$ 
\textit{EquiWing} and \textit{EquiWing-Comp} 
accordingly. 
The edge deletion has a similar effect on decreasing wing number. 
Due to the limited space, we focus on discussing the edge insertion. 
We first identify the affected edges in $G$ (Section \ref{subsection:scope_of_affected_edges}), which in-turn is used to recognize the affected super nodes in our proposed indices (Section \ref{subsection:scope_of_supernodes}) and then design a dynamic update algorithm (Section \ref{subsection:dynamic_maintenance_algorithm}). The dynamic maintenance of 
%$EquiWing-Comp$ 
\textit{EquiWing-Comp} 
is achieved using a similar approach, which is discussed at the end of the section. 

\subsection{Identification of Affected Edges} 
\label{subsection:scope_of_affected_edges}
%\paragraph{Affected Edges.} 
Let $\psi(e)$ and $\psi'(e)$ be the wing number of $e \in E$ before and after inserting/deleting an edge $e'$. 
%TODO75: fixed
Since updating $G$ from scratch is not a viable option, we confine the scope of the affected edges in $G$. We now present the following two evident observations regarding the edge insertion/deletion:
\begin{enumerate}
    \item \textbf{Observation 1}: if $e'$ is inserted into G with $\psi'(e')=a$, then $\forall e \in E$ having $\psi(e) \geq a$, $\psi'(e) = \psi(e)$ %remains unaffected.
    \item \textbf{Observation 2}: if $e'$ is deleted from G with $\psi(e')=a$, then $\forall e \in E \setminus \{e'\}$ having $\psi(e) > a$, $\psi(e)$ remains unaffected.
\end{enumerate}
The Observations $1$ and $2$ can be justified clearly as in both the cases $e'$ does not participate in \textit{$(a+1)$-wing}. 

According to Observations $1$ and $2$, if $\psi'(e')$ is known, the scope of the affected edges can be identified clearly. 
However, computing the precise $\psi'(e')$ is expensive, which may explore the entire graph and the wing number of other affected edges can be updated when computing the precise $\psi'(e')$. 
This contradicts to the intention, applying updates on the affected edges only, of using Observations $1$ and $2$. 

To speed up the update computation, instead of computing the precise $\psi'(e')$, we propose a novel upper bound for $\psi'(e')$, denoted as $\overline{\psi'(e')}$ that can be calculated at a low cost. 
Using $\overline{\psi'(e')}$, we can identify a scope that is slightly larger than the scope identified by $\overline{\psi'(e')}$, and then update the wing number for edges in the $\overline{\psi'(e')}$ identified scope only.       

%Further we provide the following lemma w.r.t change inψ(e)as follows
To derive the upper bound $\overline{\psi'(e')}$, we first introduce the definitions and lemmas below.  

%TODO76: fixed
\begin{definition}
{$k$-level butterflies:} For an edge $e(u,v)$ and $k \geq 1$, $k$-level butterflies containing $e$ are defined as $\hour_{e}^{k}=\{\hour_{uvu'v'}:\min\{\psi(e'')|e''\in \{u,u'\}\times\{v,v'\}\setminus\{(u,v)\} \} \geq k \}$. The number of butterflies in $\hour_{e}^{k}$ is denoted as $|\hour_{e}^{k}|$.
\label{def:k-level_butterflies}
\end{definition}

Before showing the lemmas below, we introduce a new notation $\delta(\cdot)$. It denotes the difference of the butterfly support/wing number of an edge after and before the insertion of an edge, i.e. $\delta(\psi(e)) = \psi'(e)-\psi(e)$ and $\delta(\hour(e)) = \hour'(e)-\hour(e)$.

%TODO77: fixed
\begin{lemma}\label{le:Bdelta}
    After the insertion of $e'$, $\forall e\in E(G)\setminus\{e'\}$, $\psi'(e)-\psi(e) \leq \Delta$, where $\Delta=\max \{\delta (\hour(e)): e \in I(e')\}$, where $I(e')$ denotes the set of edges incident to $e'$. %\delta(\hour(e))$.
\end{lemma}
\begin{proof}
%Given an edge $e$, we know that $\psi(e)\leq \hour(e)$, which also indicates $\delta(\psi(e)) \leq \delta(\hour(e))$. 
%TODO78: fixed
An upper bound for $\delta(\hour(e))$ can be used to bound $\delta(\psi(e))$. 
The increase in number of butterfly support ($\delta(\hour(e))$) for an edge can be categorised into two: (i) \textit{non-incident edges w.r.t. $e'$} - the edge $e(u,v)$, does not share any vertex with $e'$, hence it can only form $1$ new butterfly, which means $\delta(\hour(e))=1$; (ii) \textit{incident edges w.r.t. $e'$}- for an edge $e$ sharing a vertex in $e'=(u',v')$, the fourth vertex for the butterfly is not fixed, and is bounded by the common neighbors of the two non-incident vertices of $e$ and $e'$. 
%TODO79: fixed
Hence, the maximum value of $\delta(\hour(e))$ can be given as $|\Gamma(u')\cap \Gamma(u)|-1$ if the two edges incident on $v'$ or $|\Gamma(v')\cap \Gamma(v)|-1$ if the two edges incident on $u'$.
Therefore $\Delta=\max \{\{\delta (\hour(e))|e \in I(e')\},\{1\}\}$.  
%TODO80: fixed
%I(e) or I(e')?
\end{proof}

%TODO81:fixed
Using Definition~\ref{def:k-level_butterflies} and Lemma~\ref{le:Bdelta}, we can have an upper bond for $\psi'(e')$ that is let $l=\max \{k: |\hour_{e'}^{k}| \geq k \}$, then $l+\Delta$ is the upper bound. 
%TODO82: fixed
This upper bound is loose since it assumes there are $l$ number of $l$-level butterflies before the insertion of $e'$ and assumes that the insertion makes the wing number of every edge in the $l$-level butterflies increase by $\Delta$, which is quite unlikely. %In fact, before the insertion, there may not be such many $l$-level butterflies and the butterfly support increase induced by the insertion could be less than $\Delta$. 
A tighter upper bound for $\psi'(e')$ that we use is based on the lemma below.   

%TODO83: fixed
\begin{lemma}
If an edge $e'(u',v')$ is inserted into $G$, then $\psi'(e')\leq h$ holds, where $h=\max \{k: |\hour_{e'}^{k-\Delta}| \geq k \}$.
\label{lem:bounds}
\end{lemma}
%\begin{proof}

%\begin{lemma}
%If an edge $e'(u',v')$ is inserted into a graph, then $\psi'(e')$ satisfies $l \leq \psi'(e') \leq h$ and $h-l \leq \Delta$, where $l=\max \{k: |\hour_{e'}^{k}| \geq k \}$, $h=\max \{k: |\hour_{e'}^{k-\Delta}| \geq k \}$.
%\label{lem:bounds}
%\end{lemma}
\begin{proof}
We prove $\psi'(e')\leq h$ by contradiction. Let $\psi'(e')=a > h$, then there exists a $a$-wing cohesive subgraph where all the edges have wing number no less than $a$, and $e'$ would be contained by at least $a$ butterflies. Since the wing number of all the edges except $e'$ can increase at most by $\Delta$ after insertion, we have $|\hour_{e'}^{a-\Delta}| \geq a$ before insertion. This also implies that $h \geq a$ as per the definition of $h$, which is a contradiction. 
\end{proof}

%Further, implementing these observations and lemma we need to compute $\psi'(e')$, which is costly. Hence we propose an upper bound for $\psi'(e')$, denoted as $\overline{\psi'(e')}$, and update the Observation $1$ as follows:
%\begin{enumerate}
%    \item \textbf{Observation 1'}: if $e'$ is inserted into G with $\psi'(e')=a$, then $\forall e \in E$ having $\psi(e) \geq \overline{\psi'(e')} \geq a $, $\psi(e)$ remains unaffected.
%\end{enumerate}
%We then define \textit{$k$-level butterflies} and use $\Delta$ for an edge to estimate $\overline{\psi'(e')}$.

 % and having $k$-level butterflies. %| $e \in k $-wing$\}$. 
% However, the maximum value of $\delta$ could be $\{|\Gamma(u')\cap \Gamma(u)|-1$ or $|\Gamma(v')\cap \Gamma(v)|-1\}$ which is a loose bound. We define \textit{pseudo neighbors} to propose a tighter bound for calculating $\delta$.
% \begin{definition}
% For a given vertex $i \in (U \cup V)$, the pseudo neighbor of $i$ are given and denoted as $pn(i)=\bigcup \limits_{j \in \Gamma(i)}^{}\Gamma(j)$. 
% \label{def:pseudo_neighbor}
% \end{definition}
% % is the maximum change in the support of an edge $e$, where $e \in \hour_{e'}^{k}$. 
% Using the Definition \ref{def:pseudo_neighbor}, we propose the tighter bound for $\delta$ as $\delta(\hour(e))=\{|ps(u')\cap ps(v')|-1\}$.
With the aid of Lemma \ref{lem:bounds}, the upper bound $\overline{\psi'(e')}$ is given below.

\begin{cor}\label{cor:upperbound}
$\overline{\psi'(e')}=\max \{k:|\hour_{e'}^{k-\Delta}| \geq k\}$ is an upper bound of $\psi'(e')$.
\end{cor}
% \begin{example}
% Fig. \ref{fig:change in support} displays how the wing numbers are affected from the insertion of edges. In Fig. \ref{fig:change in support a} edge $e'(v_5,u_2)$ is inserted in the graph $G$. The change in wing number can be observed in Fig. \ref{fig:change in support b}.
% In Fig. \ref{fig:change in support a} we have  $l=\hour_{(v_5,u_2)}^4= \{\hour_{v_5u_2v_1u_3},\hour_{v_5u_2v_2u_3},\hour_{v_5u_2v_3u_3}$, $\hour_{v_5u_2v_4u_3},\hour_{v_5u_2v_1u_4}$, $\hour_{v_5u_2v_2u_4}$,$\hour_{v_5u_2v_3u_4}$,$\hour_{v_5u_2v_4u_4}\}$, $l=4$, $\Delta=4$ for the edge $\delta(v_5,u_3)=|\Gamma(u3) \cap \Gamma(u2)|-1=5-1=4$ and $\overline{\psi(e')}=8$. Thus, $4 \leq \psi(e')\leq 8$. Similarly, for Fig. \ref{fig:change in support b}, when the edge $e'(v_5,u_1)$ is inserted $l=8$, $\Delta=4$ and $\overline{\psi(e')}=12$. Hence, $8 \leq \psi(e')\leq 12$.
% \end{example}

Using $\overline{\psi'(e')}$, the relaxed scope of the affected edges can be derived, i.e., $\forall e$ such that $\psi(e)$ $\le$ $\overline{\psi'(e')}$.% with wing numbers

Next we propose the theoretical findings that can help us further exclude the edges in the relaxed scope that do not need to perform update computations. The intuition is given below. According to Lemma~\ref{le:Bdelta}, although the insertion of $e'$ may increase the wing number of some edges by $\Delta$, it may also increase the wing number of some edges at most by $1$ as well. This motivates us to study the identification of edges such that they even cannot increase the wing number by $1$ after inserting $e'$ and we exclude these edges for the update computations. We produce the lemma below. 
%Due to the insertion of edge $e'$, there are two reasons for edge $e\in E \setminus \{e'\}$ to change the wing number, i.e. $\psi'(e) \neq \psi(e)$: (i) $e'$ forms a butterfly with $e$. (ii) the edges of the butterflies in which $e$ lies have changed their wing number. We produce Lemma \ref{lem:k connected} for the insertion of an edge.

%TODO: many problems, 1)after an insertion, the wing number changes upto a+$\delta(\hour(e_0))$ instead of exact. 
\begin{lemma}
\label{lemma:3}
If an edge $e'(u',v')$ is inserted into $G$, we first calculate the value of $\overline{\psi'(e')}$. We then assign $\psi(e')=\max \{k: |\hour_{e'}^{k}| \geq k \}$, and for all the edges $e_0(u_{0},v_{0}) \in E \cup \{e'\}$ with $\psi(e_0)=a \leq \overline{\psi'(e')}$, the wing number of $e_0$ may be updated as $\psi'(e_0)\ge a+1$, if and only if:
\begin{enumerate}
    \item A new butterfly with edges of $\{u',v'\}$ $\times$ $\{u_{0},v_{0}\}$ is formed, and $\min\{\psi((u',v_{0})),\psi((v',u_{0})), \psi((u',v')) \} > a-\Delta$; or
    \item For any $e_0(u_{0},v_{0})$, $\exists  \hour_{u_{0}v_{0}uv}$ such that $ \min$ $\{$$\psi((v_{0},u)),$ $\psi((u,v)),$ $\psi((v,u_{0}))$$\}$ $>$ $a-\Delta$ holds. % comma in butterfly representation?
\end{enumerate}
%\begin{enumerate}
%    \item A new butterfly with edges of $\{u',v'\}$ $\times$ ${u_{0},v_{0}}$ is formed, and $\min\{\psi(e'),\psi(x),\psi(x')\} \geq a+\delta(\hour(e_0)$; or
%    \item For an edge $e_0$, \exists  $\hour_{ee_0xx'}$ s.t. $\min\{\psi(e),\psi(x),\psi(x')\} = a$ holds.
%\end{enumerate}

%$e \xLeftrightarrow{\hour_a} e'$. 
\label{lem:k connected}
\end{lemma}
\begin{proof}
Let there be an edge $e_0(u_0,v_0)$ with $\psi(e_0)=a$ and $\psi'(e_0)$ $=$ $a+ 1$. 
%TODO84: fixed
The easiest case to satisfy is below. Before the insertion, $e_{0}$ only needs to involve $(a+1)$ number of $(a-\Delta+1)$-level butterflies and after the insertion, the butterfly support of all the edges in these 
%TODO85: fixed
$(a-\Delta+1)$-level butterfly may increase by $\Delta$, which makes them become $(a+1)$-level butterflies and 
%TODO86: fixed
increases the wing number for edge $e_0$ from $a$ to $a+1$. 
Then we prove that only the edges satisfying cases (1) and (2) can increase their wing number.

%TODO87: fixed
For case (1), due to the insertion of $e'(u',v')$, the butterfly support for $(u',v_{0})$, $(v',u_{0})$, and $(u',v')$ could increase by $\Delta$. As such, after the insertion $\min$$\{$$\psi((u',v_{0})),$ $\psi((v',u_{0})),$$\psi((u',v'))$$\}$ can be increased to a value no less than $a+1$ and it is also possible that there are at least $(a+1)$ number of $(a+1)$-level butterflies, which makes %the wing number of $e_{0}$ at least $a+1$ after the insertion. ( Replaced wing number after insertion with notation, to save a line)
the value of $\psi'(e_0)$ at least $a+1$. 

%TODO88: fixed
For case (2), since $\min$$\{$$\psi((u',v_{0})),$ $\psi((v',u_{0})),$$\psi((u',v'))$$\}$ could be increased by $\Delta$ after the insertion, $\min$ $\{$ $\psi((u',v_{0})),$$\psi((v',u_{0})),$ $\psi((u',v'))$$\}$ can be increased to a value in $[a+1,\overline{\psi'(e')})$ and it is also possible that there would be at least $i$ number of $i$-level butterflies ($a+1\le i< \overline{\psi'(e')}$), which makes %the wing number of $e_{0}$ at least $a+1$ after the insertion. ( Replaced wing number after insertion with notation)
the value of $\psi'(e_0)$ at least $a+1$.

Next consider an edge $e_{0}(u_{0},v_{0}) \in E \cup \{e'\}$ with $\psi(e_{0})=a$, if it does not satisfy case (1), there is no new formed butterfly containing $e_{0}$ to account for the increase in the wing number. Hence, it is impossible to have $\psi'(e_{0})\geq a$. %will it be > or >=.

%TODO89: fixed
If it does not satisfy case (2), then the existing butterfly $\hour_{u_{0}v_{0}uv}$ is with either $\min$$\{$$\psi'((u_{0},v)),$ $\psi'((v,u)),$$\psi'((u,v_{0}))$$\}$ $<$ $a$ $-$ $\Delta$ $+1$ or $\min$$\{$$\psi'((u_{0},v))$,$\psi'((v,u))$,$\psi'((u,v_{0}))$$\}$  $\geq$ $\overline{\psi'(e')}$. After the insertion we respectively have $\min$$\{$$\psi'((u_{0},v)),$ $\psi'((v,u)),$$\psi'((u,v_{0}))$$\}$ $<$ $a+1$ or $\min$$\{$$\psi'((u_{0},v)),$ $\psi'((v,u)),$ $\psi'((u,v_{0}))$$\}$ $\geq$ $\overline{\psi'(e')}+\Delta$.
For the first situation, after the insertion, there would be no sufficient number of $(a+1)$-level butterflies for $e_{0}$ for increasing its wing number. 
The second situation is out of the scope of this lemma and the wing number would not be affected due to the insertion.  
Thus if $e_{0}$ satisfies neither case (1) nor (2), it is impossible to have $\psi'(e_{0}) \ge a+1$. 
\end{proof}
Based on Lemma \ref{lemma:3} and the proposed upper bound, we refine the affected edge scope due to edge insertion/deletion as follows.

\noindent\textbf{Refined scope of the affected edges.} Using Corollary \ref{cor:upperbound} and Lemma \ref{lem:k connected}, we now summarize the edges whose wing number could be affected by the insertion/deletion of the edge $e'$:
\begin{enumerate}
    \item \textbf{Insertion case.} For $e \in E \cup \{e'\}$ with $\psi(e) < \overline{\psi'(e')}$, if $e$ and $e'$ form a new $k$-butterfly, or $e$ is connected to $e'$ via a series of adjacent $k$-butterflies i.e. $e \xLeftrightarrow{k} e'$, then $e$ may have $\psi'(e) \ge \psi(e)$.
    %TODO89: fixed
    %TODO90: fixed
    \item \textbf{Deletion case.} For $e \in E \setminus \{e'\}$ with $\psi(e) \leq \psi(e')$, if $e$ and $e'$ belong to a $k$-butterfly, or $e$ is connected to $e'$ via a series of adjacent $k$-butterflies i.e. $e \xLeftrightarrow{k} e'$ before the deletion, then $e$ may have $\psi'(e) \leq \psi(e)$. 
\end{enumerate}

\begin{figure}[]
\centering
    
\begin{subfigure}{0.3\textwidth}
    
    \scalebox{0.3}{
    \begin{tikzpicture}
\foreach \name/ \x in {v_1/1,v_2/3,v_3/5,v_4/7,v_5/9,v_6/11,v_7/13,v_8/15}
\node (\name) at ( \x,6) [circle,draw=blue!20,fill=blue!20,very thick,inner sep=7pt] {\Huge $\name$};

\foreach \name/ \x in {u_1/2,u_2/4,u_3/6,u_4/8,u_5/10,u_6/12,u_7/14}
\node (\name) at (\x,3) [circle,draw= blue!20,fill=blue!20,very thick,inner sep=7pt] {\Huge $\name$};
%\node (u3) at (5,6) [circle,draw=red!50,fill=red!20,very thick,inner sep=7pt] {$u_3$};
\path[]
(v_1) 
edge  [very thick](u_1)
edge  [very thick](u_2)
(v_2) edge  [very thick] (u_1)
edge [very thick,green] (u_2)
edge  [very thick,red] (u_3) 
edge  [very thick,red] (u_4)
(v_3) edge [very thick,red] (u_4)
edge  [very thick,green](u_2)
edge  [very thick,red](u_3)
(v_4) 
edge  [very thick,red](u_4)
edge  [very thick,red](u_3)
edge  [very thick,dashed](u_6)
(v_5) edge  [very thick,red](u_3)
edge  [very thick,red](u_4)
edge  [very thick,red](u_5)
edge  [very thick,red](u_6)
(v_6) edge  [very thick,blue](u_5)
edge  [very thick,blue](u_6)
edge  [very thick,blue](u_7)
edge  [very thick,green](u_4)
(v_7) edge  [very thick,blue](u_5)
edge  [very thick,blue](u_6)
edge  [very thick,blue](u_7)
(v_8) edge  [very thick,blue](u_5)
edge  [very thick,blue](u_6)
edge  [very thick,blue](u_7);
\end{tikzpicture}%
}
\caption{Edge Insertion}
\label{subfig:edge_insert}
\end{subfigure} 
%\hspace{1em}
\begin{subfigure}{0.15\textwidth}
%\centering
\scalebox{0.3}{

\begin{tikzpicture}
\node (1) at (0,3) [circle,draw=black!50,very thick,inner sep=5pt] { \Huge $\nu_1$};
\node (2) at (-1,4)  {\huge $\psi(\nu_1)=1$};

\node (3) at (-1,0)  {\huge $\psi(\nu_2)=2$};
\node (4) at (0,1) [circle,draw=black!50,very thick,inner sep=5pt] {\Huge $\nu_2$};
 
\node (5) at (2.5,3) [circle,draw=black!50,very thick,inner sep=5pt] {\Huge $\nu_3$};
\node (6) at (2.5,4)  {\huge $\psi(\nu_3)=3$};

\node (7) at (2.5,1) [circle,draw=black!50,very thick,inner sep=5pt] { \Huge $\nu_4$};
\node (8) at (2.5,0)  {\huge $\psi(\nu_4)=4$};

\path[]
(1) edge [very thick] (4)
(4)edge [very thick] (5)
(5)edge [very thick] (7);

\end{tikzpicture}%
}
\caption{Updated $EW$ }
\label{subfig:updated_equiwing}
\end{subfigure}

\caption{Inserted new edge $(v_4,u_6)$ and the updated $EW$ }
\label{fig:inserted_edge}
\vspace{-10pt}
\end{figure}
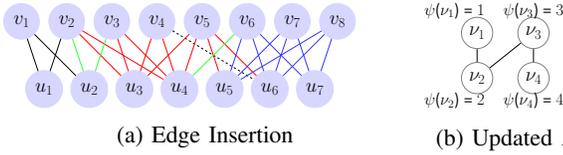

\subsection{Identification of Affected Super Nodes}
\label{subsection:scope_of_supernodes}
Since all the $k$-butterfly equivalent edges in the graph are grouped to form super nodes in the EquiWing index, the scope of the affected super nodes can be derived using the scope of the affected edges. We now propose the following theorem indicating the affected super nodes in \textit{EquiWing}. 

%TODO91: 
\begin{theorem}
Given an inserted edge $e'(u',v')$ and a $k$-butterfly $\hour_{u'v'uv}$ where $k < \overline{\psi'(e')}$, then the following super nodes in \textit{EquiWing} may be updated: $\{ \nu | \nu \in \raisebox{1.5pt}{$\chi$}, e\in \nu, \psi(e)=k\}$.
\label{theorem:insert}
\end{theorem}
Theorem \ref{theorem:insert} can be easily proved using the scope of affected edges for the insertion. Since, only the \textit{$k$-butterfly connected} edges are affected due to the insertion of $e'$, which is same as the edges in a super nodes. Hence, when a new edge $e'$ is inserted to $G$, all the edges with a potential wing number increase, except $e'$, are contained in the affected super nodes of \textit{EquiWing}. Therefore, we avoid re-examining the original graph $G$ to find the affected edges and hence save significant computational cost.
\begin{example}
We insert a new edge $e'(u_6,v_4)$ in $G$ as presented in Fig. \ref{subfig:edge_insert}, and $\overline{\psi'(u_6,v_4)}=4$. We examine the butterflies $\hour_{u_6v_4u_4v_5},$ $\hour_{u_6v_4u_3v_5},$ and $\hour_{u_6v_4u_4v_6}$ as it includes $e'$. Since, $\psi(u_6,v_6)=4 \nless 4$, it is not updated including the edges $k$-butterfly connected with $(u_6,v_6)$. However, the other edges $(u_3,v_4),(u_3,v_5),(u_4,v_5),(u_6,v_5),(u_4,v_4),$ and $(u_4,v_6)$, have the $\psi(e) < 4$. We recognize that $\nu_3,\nu_4$, and $\nu_5$ (Fig. \ref{fig:EquiWing}) include these edges, so they are considered to be affected. As a result, all the other edges within these super nodes including the inserted edge $e'$, will be re-examined as their wing number may increase, based on Theorem \ref{theorem:insert}. 
%TODO92: 
%Fig. \ref{subfig:updated_equiwing} displays the updated \textit{EquiWing} after the insertion of $e'$ in $G$.
\end{example}

\begin{algorithm}[t]
\small
	\SetKwInOut{Input}{Input}
	\SetKwInOut{Output}{Output}
	\Input{The affected edge set $E'$, the affected super node set $\raisebox{1.5pt}{$\chi'$}$, and the EquiWing index $EW$} 
	\Output{The updated EquiWing $EW$} 
	\SetKwFunction{FMain}{Update}
    \SetKwProg{Fn}{Function}{}{}
    \Fn{\FMain{$\raisebox{1.5pt}{$\chi'$}$,$E'$,$EW$}}{
        $EW' \longleftarrow EW$\\
        \ForEach{$\nu \in \raisebox{1.5pt}{$\chi'$}$}{
            $EW' \longleftarrow EW' - \{\nu\}$
        }
        \ForEach{$e(u,v) \in E'$}{
            $U'\longleftarrow U' \cup u$,
            $V'\longleftarrow V' \cup v$
            }
        $E'\longleftarrow$  $Update_K(G'(U',V'));$ \tcc{Update wing number of the affected edges E'.}
        $\delta_{EW}\longleftarrow Index$  $ Construction(G'(U',V'));$ \tcc{Call Algorithm \ref{Algorithm:EquiWIng}.}
        $EW' \longleftarrow EW' \cup \delta_{EW}$\\
        \ForEach{$(\nu,\mu) \in \Upsilon'_{EW}$}{
        \If{$\psi(\nu)=\psi(\mu) $}{
        \tcc{merge the two super nodes $\nu$ and $\mu$}}
        }
   
        \textbf{return} $EW(\raisebox{1.5pt}{$\chi$},\Upsilon)$;
}

\caption{Dynamic maintenance of EquiWing.}
\label{Algorithm:dynamic_update}

\end{algorithm}

Likewise, we propose the theorem for the deletion scenario indicating the affected super nodes of \textit{EquiWing}:
\begin{theorem}
Given an edge $e'(u',v')$ deleted from $G$, and $e' \in \nu$ where $\nu \in \raisebox{1.5pt}{$\chi$}$. Then the following super nodes in \textit{EquiWing} may be updated: $\{\nu\} \cup \{\mu | \psi(\mu) < \psi(e'),\mu \in \raisebox{1.5pt}{$\chi$}, (\nu,\mu) \in \Upsilon\}$.
\label{theorem:delete}
\end{theorem}
The proof of Theorem \ref{theorem:delete} can be similarly accomplished as of Theorem \ref{theorem:insert}, hence, we omit the its discussion for concision.% the sake of brevity.

% using the the scope of affected edges in deletion. According to Theorem \ref{theorem:delete}, when a new edge $e'$ is deleted from $G$, all the edges with a potential wing number decrease are contained in the affected super nodes of \textit{EquiWing}. Therefore, we save a significant computational cost of exploring the whole graph by restricting the affected edges to a small region.
% \begin{example}
% We delete an edge $e'(v_3, u_4)$ in $G$ as presented in Fig. \ref{subfig:edge_delete}, and $\psi(v_3,u_4)=3$. The edge $e' \in \nu_3$ in Fig. \ref{fig:EquiWing}, all the edges within $\nu_4$ are affected. Further $\nu_2$ and $\nu_3$ are connected via super edges and $\psi(\nu_2)=\psi(\nu_3)=2 \leq 3$, hence these two nodes are also affected. As a result, all the edges within these super nodes are re-examined as their wing number may decrease, based on Theorem \ref{theorem:insert}. Fig. \ref{subfig:updated_equiwing} displays the updated EquiWing after deleting $e'$ from $G$.
% \end{example}

\subsection{Dynamic Maintenance Algorithm}
\label{subsection:dynamic_maintenance_algorithm}

\noindent\textbf{EquiWing update}. 
After an edge insertion, let us denote the set of affected edges in $G$ as $E'$ and the set of affected super nodes in the 
\textit{EquiWing} 
index as $\raisebox{1.5pt}{$\chi$'}$. 
We can easily identify $E'$ and $\raisebox{1.5pt}{$\chi$'}$ according the theoretical findings proposed in Sections \ref{subsection:scope_of_affected_edges} and \ref{subsection:scope_of_supernodes}. 
Using $E'$ and $\raisebox{1.5pt}{$\chi$'}$, the 
\textit{EquiWing} 
index can be updated by Algorithm \ref{Algorithm:dynamic_update} as follows. 

\noindent\textit{Super nodes update}. We build a new 
%$EquiWing$ 
\textit{EquiWing} 
index $\delta_{EW}$ only for the $E'$ induced subgraph of $G$ using Algorithm \ref{Algorithm:EquiWIng}, lines 7 to 8 in Algorithm \ref{Algorithm:dynamic_update}.
%AMAN
%As we already updated the wing number for all the edge in $E'$ (line $7$), $bitruss\_Decomposition$ in Algorithm \ref{Algorithm:EquiWIng} is not required.  
A super node in $\delta_{EW}$ can be merged into an existing unaffected super node in the index, which is done by lines $10$ to $12$.

\noindent\textit{Super edge update}. We identify a subgraph $G'\subseteq G$, which is the subgraph induced by vertices contained in $E'$. 
The purpose for identifying $G'$ is to preserve the butterfly connectivity between the affected edges and the unaffected edges, which is used to update the super edges in the %$EquiWing$ 
\textit{EquiWing} 
index. 
Thanks to $G'$, we can safely remove $\raisebox{1.5pt}{$\chi$'}$ from the index (lines $3$ to $4$) since after the super node update, the newly computed super nodes in $\delta_{EW}$ can be linked to the index easily.

\begin{example}
For an inserted edge $e'=(u_6,v_4)$ in Fig. \ref{subfig:edge_insert}, the affected super nodes are $\nu _3, \nu_4$ and $\nu_5$. Algorithm \ref{Algorithm:dynamic_update} computes $\delta_{EW}$ for the affected edges in the affected super nodes, resulting into two new super nodes $\nu_2$ and $\nu_3$. $\delta_{EW}$ is then appended to the existing index with unaffected super nodes $\nu_1$ and $\nu_4$. Fig. \ref{subfig:updated_equiwing} displays the updated \textit{EquiWing}.
\end{example}
%TODO95: fixed
\noindent\textbf{EquiWing-Comp update}. The dynamic update for the \textit{EquiWing-Comp} is achieved by compressing the updated \textit{EquiWing}. The running time required for the compression is negligible, which will be shown later in the experiments.% later in this paper.

\section{Experiments}
\label{section:experiments}
In this section, we perform the experimental analysis for our proposed techniques on real-world datasets. \\
\noindent\textbf{Experiment setup.} All the experiments were conducted on Eclipse IDE, deployed on the platform $64$x Intel(R)Core(TM) $i7$-$1065G7$ with CPU frequency $1.50$GHz and $16$ GB RAM, running Windows 10 Home operating system. %To perform a fair and comprehensive comparison between the algorithms, we only clocked the runtime of each algorithm. The experiments were performed without using any kind of parallelism. 
All the algorithms were implemented in Java. 
\\
\noindent\textbf{Algorithms.} We have implemented the following three algorithms for our experiments:
\begin{itemize}
    \item \textbf{BaseLine}: it is the implementation for $k$-wing search with wing number.
    \item \textbf{EquiWing ($EW$)}: it exploits the index \textit{EquiWing} to perform the subgraph search using Algorithm \ref{Algorithm:Community_EquiWing}.
    \item \textbf{EquiWing-Comp} ($EW\raisebox{0.5pt}{-}{C}$): this approach utilizes the index \textit{EquiWing-Comp} that is created by compressing \textit{EquiWing} and performs the search. % using Algorithm \ref{Algorithm:Community_EquiWing_2}. 
\end{itemize}
%Since %the paper addresses the problem of the \textit{cohesive bipartite subgraph} and does not address the problem of wing decomposition, hence 
Since the bitruss decomposition algorithm used for all the three algorithms is \cite{DBLP:conf/icde/Wang0Q0020}, the running time of all the algorithms is independent of the bitruss decomposition algorithm and it is replaceable if in future a better option is available. For our experiments, we consider that a query vertex could be in either $U$ or $V$ for all the algorithms. 

\noindent\textbf{Datasets.} $4$ real-world datasets, obtained from the KONECT repository \cite{kunegis2013konect}, are used to evaluate our algorithms. Table \ref{table:datasets} represents the characteristics of the datasets, including the number of edges, the numbers of vertices for $U$ and $V$, %maximum degree ($d_{max}$) 
and the maximum wing number ($k_{max}$) for the datasets. %The evaluation process is divided into four sub processes: (i) evaluation of the index constructions for indices; (ii) query processing, where we evaluate the runtime for finding $k$-wings for a query vertex \textit{q}; (iii) dynamic maintenance of the indices; (iv) case study on \textit{Unicode} to demonstrate the effectiveness of the personalized $k$-wing model. 

\vspace{-5pt}
\subsection{Index Construction}
\label{subsection:index_constructiion}
We perform the experiments to evaluate the efficiency of the two indexing schemes: \textit{EquiWing} ($EW$) and \textit{EquiWing-Comp} ($EW\raisebox{0.5pt}{-}{C}$) in terms of index size, number of super nodes in the indices and index construction time in Table \ref{table:index_construction}. %The index construction time is also recorded for both of the indices. 
%We notice that there is a huge amount of nodes that can be combined in the \textit{EquiWing} and combining them reduces the size of the index to a great extent.

We observe in Table \ref{table:index_construction} that $EW\raisebox{0.5pt}{-}{C}$ is more compact as compared to $EW$, i.e., there is a large number of super nodes in $EW$ that can be combined, which reduces the number of the super nodes contained in $EW\raisebox{0.5pt}{-}{C}$ to a great extent. 
Further, we use the \textit{Compression Ratio} ($C_R$) (the ratio of the number of super nodes in $EW$ and the number of compressed super nodes in $EW\raisebox{0.5pt}{-}{C}$) to measure the compactness provided by $EW\raisebox{0.5pt}{-}{C}$. The value of $C_R$ increases as the size of the input bipartite graph increases and $C_R$ is upto $118.07$ for \textit{Stackoverflow}. %The compactness can also be observed by the reduced size of the \textit{EW-C}. 
% The size of \textit{EW-C} is also smaller as compared to the original bipartite graph size as shown in Table \ref{table:index_construction} upto $4$ times for \textit{Stackoverflow}. 
%TODO:
The size of $EW\raisebox{0.5pt}{-}{C}$ is also smaller as compared to $EW$ as shown in Table \ref{table:index_construction} upto $8$ times for \textit{YouTube}. 
%We can conclude that \textit{EC-W} takes more time to construct than \textit{EW}, however, the resulting trade-off is for the compactness of \textit{EW-C}. 
Later, this compactness of $EW\raisebox{0.5pt}{-}{C}$ leads to the efficient query processing.

\vspace{-5pt}
\subsection{Query Processing}
\label{subsection:query_processing}
\begin{table}[t]
\small
\caption{Dataset characteristics }\label{table:datasets}
\centering
\scalebox{.65}{
% \begin{tabular}[]{|m{5.5em}|m{2.3em}|m{2.3em}|m{2.3em}|m{2.2em}|m{3em}|m{3em}|}
% \hline
% Datasets &  \centering  $ |E|$ & \centering V &\centering U  & $d_{max}$ & $k_{max}$  \\ \hline
% Producer & 207K & 188K & 49K  & 512 & 219 \\ \hline
% Record\_label & 233K & 186K & 168K  & 7446 & 497 \\ \hline
% YouTube & 293K & 124K & 94K  & 7591 & 1102 \\ \hline
% Stackoverflow & 1.3M & 641K & 545K  & 6119 & 1118 \\ \hline
% \end{tabular}
\begin{tabular}[]{|m{5.5em}|m{2.3em}|m{2.3em}|m{2.3em}|m{2.2em}|m{3em}|}
\hline
Datasets &  \centering  $ |E|$ & \centering V & \centering U   & $k_{max}$  \\ \hline
Producer & 207K & 188K & 49K   & 219 \\ \hline
Record\_label & 233K & 186K & 168K   & 497 \\ \hline
YouTube & 293K & 124K & 94K   & 1102 \\ \hline
Stackoverflow & 1.3M & 641K & 545K   & 1118 \\ \hline
\end{tabular}

}
\vspace{-10pt}
\end{table}

\begin{figure*}
\hspace{-0.5em}
\begin{subfigure}{0.2\textwidth}
\includegraphics[scale=0.115]{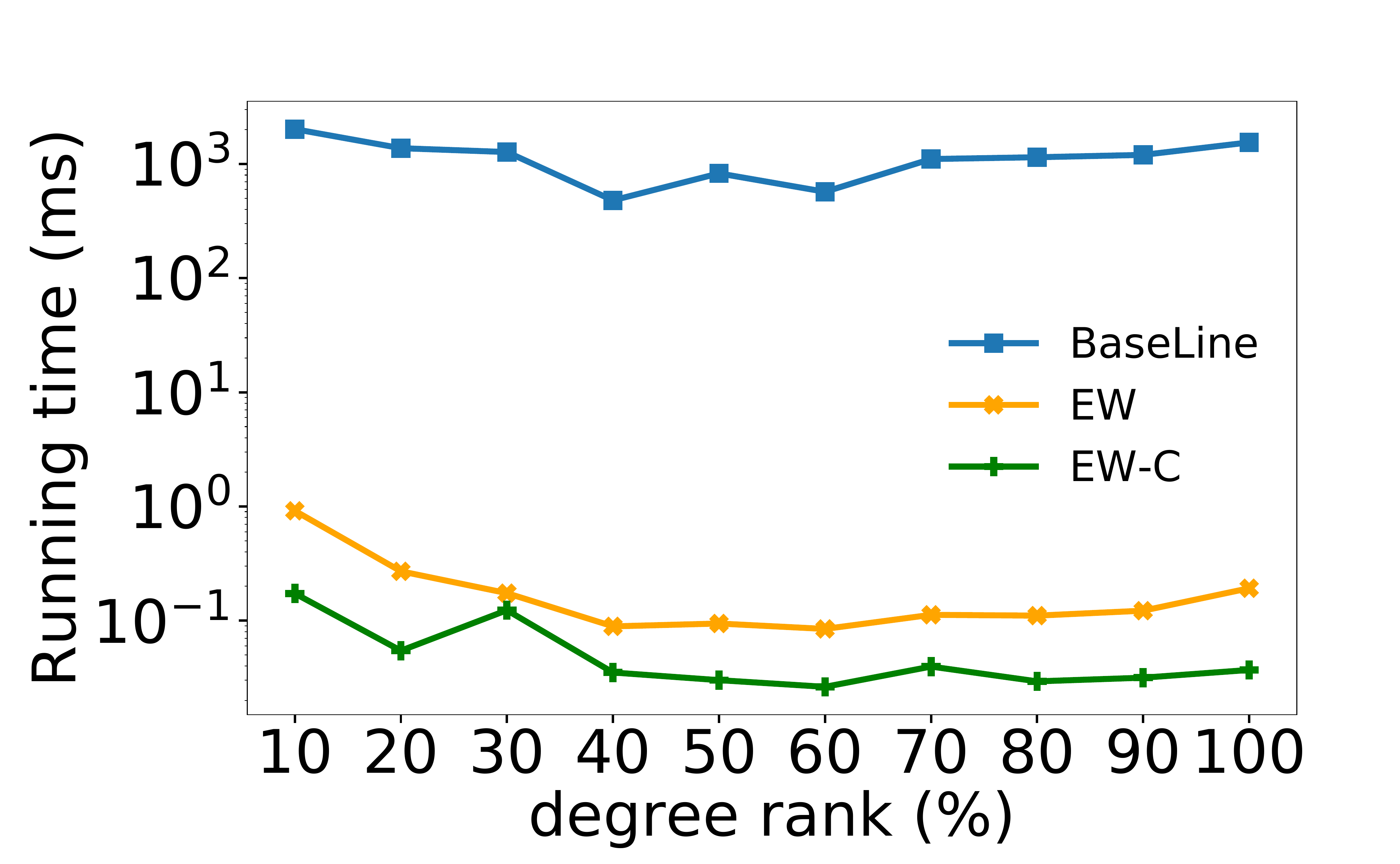}
\subcaption{Producer}
\end{subfigure}
\hspace{2em}
\begin{subfigure}{0.2\textwidth}
\includegraphics[scale=0.115]{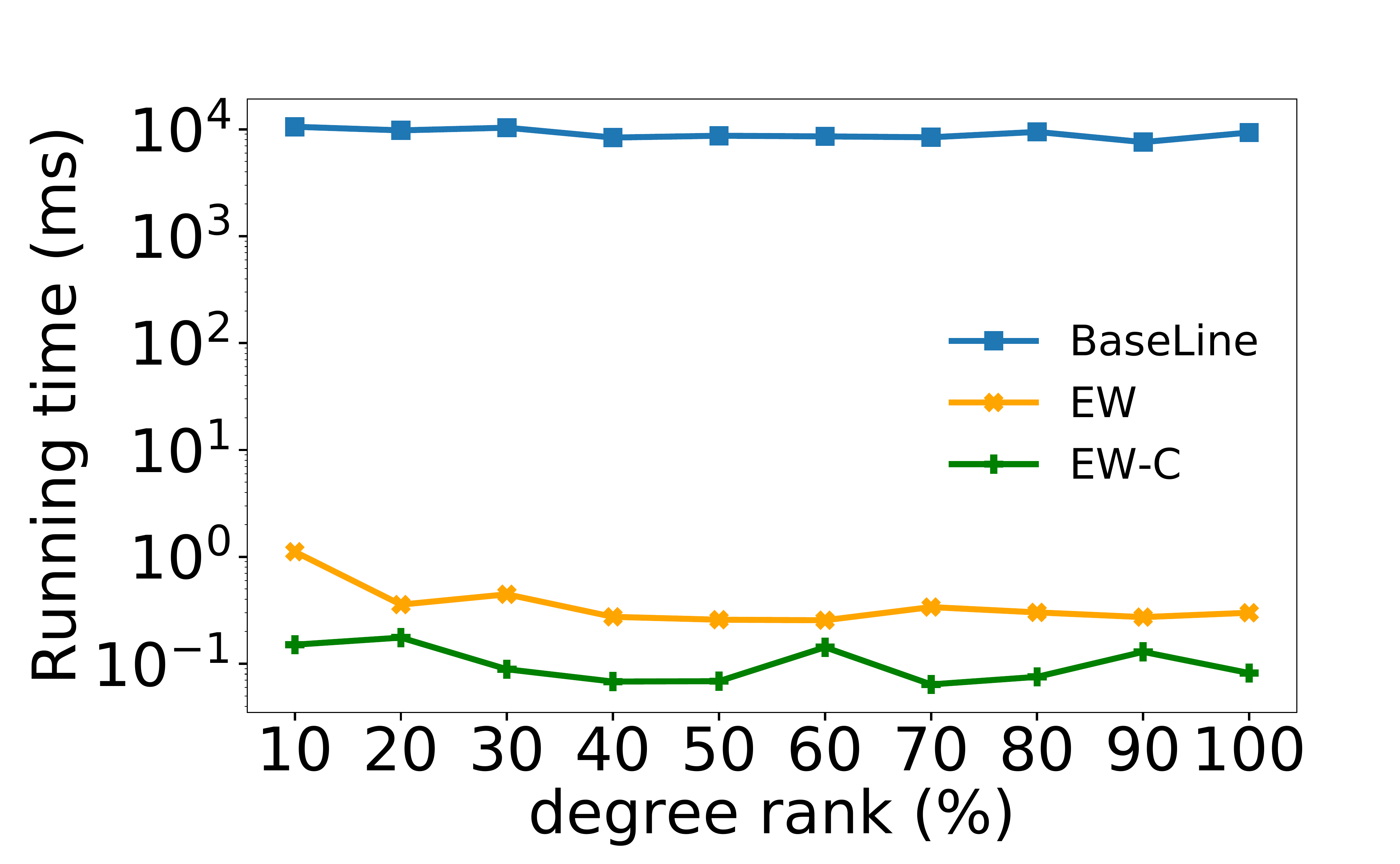}
\subcaption{Record\_label}
\end{subfigure}
\hspace{2em}
\begin{subfigure}{0.2\textwidth}
\includegraphics[scale=0.115]{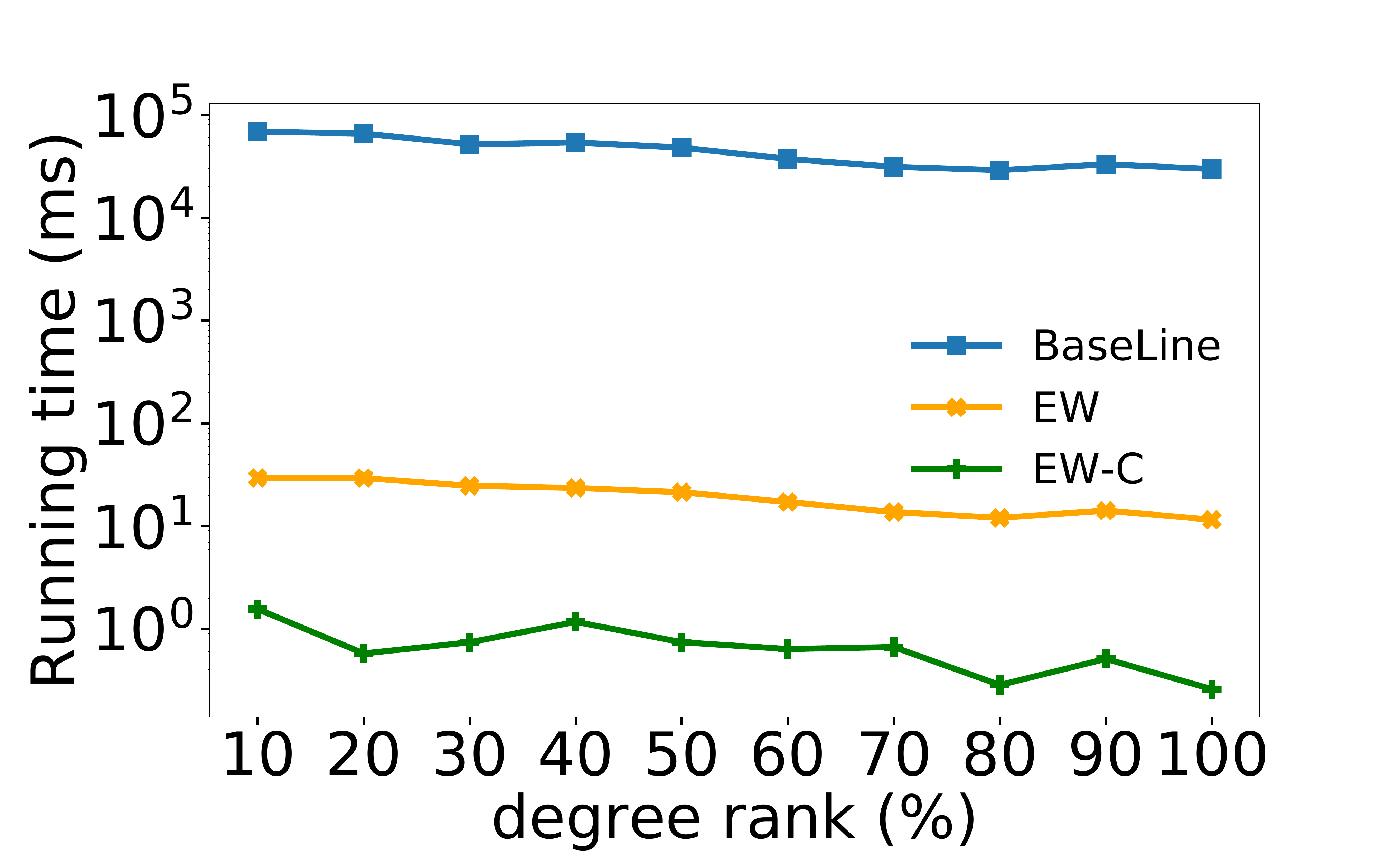}
\subcaption{YouTube}
\end{subfigure}
\hspace{2em}
\begin{subfigure}{0.2\textwidth}
\includegraphics[scale=0.115]{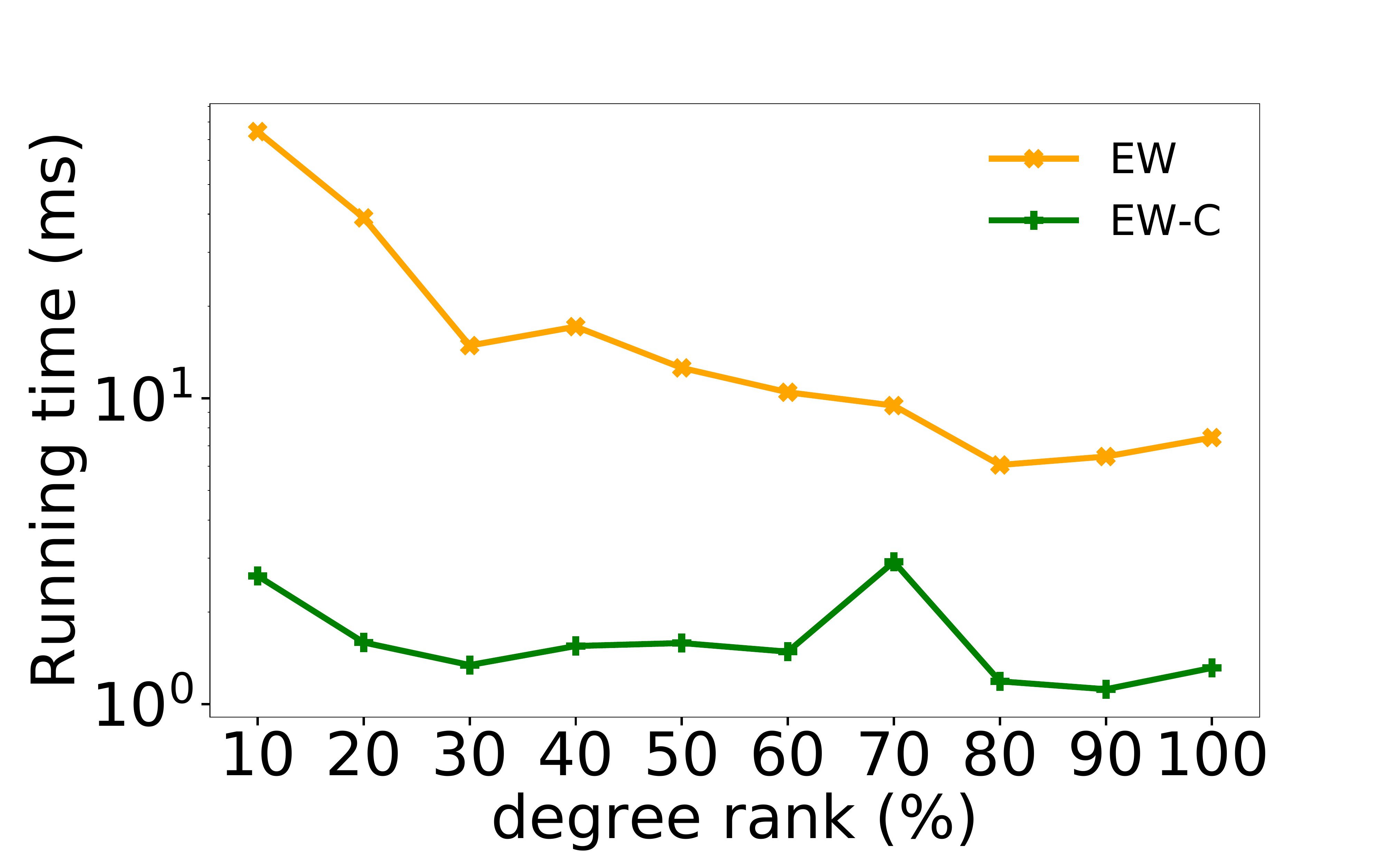}
\subcaption{Stackoverflow}
\label{subfig:stack_query}
\end{subfigure}
\caption{Cohesive subgraph performance on different vertex degree percentile on real-world datasets}
\label{fig:exp}
\vspace{-12pt}
\end{figure*}
\begin{figure*}
\hspace{-1em}
\begin{subfigure}{0.2\textwidth}
\includegraphics[scale=0.11]{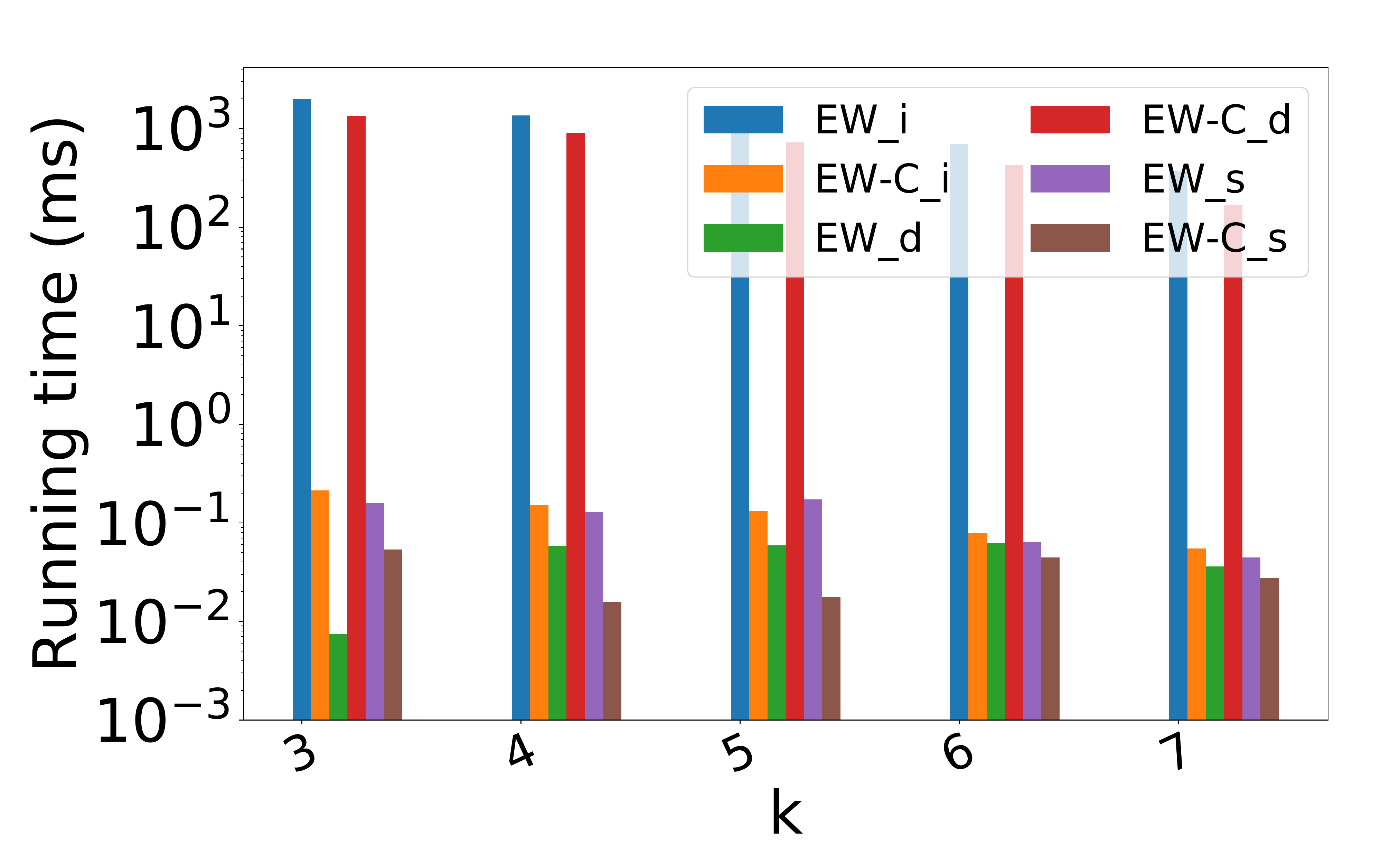}
\subcaption{Producer}
\end{subfigure}
\hspace{2em}
\begin{subfigure}{0.2\textwidth}
\includegraphics[scale=0.11]{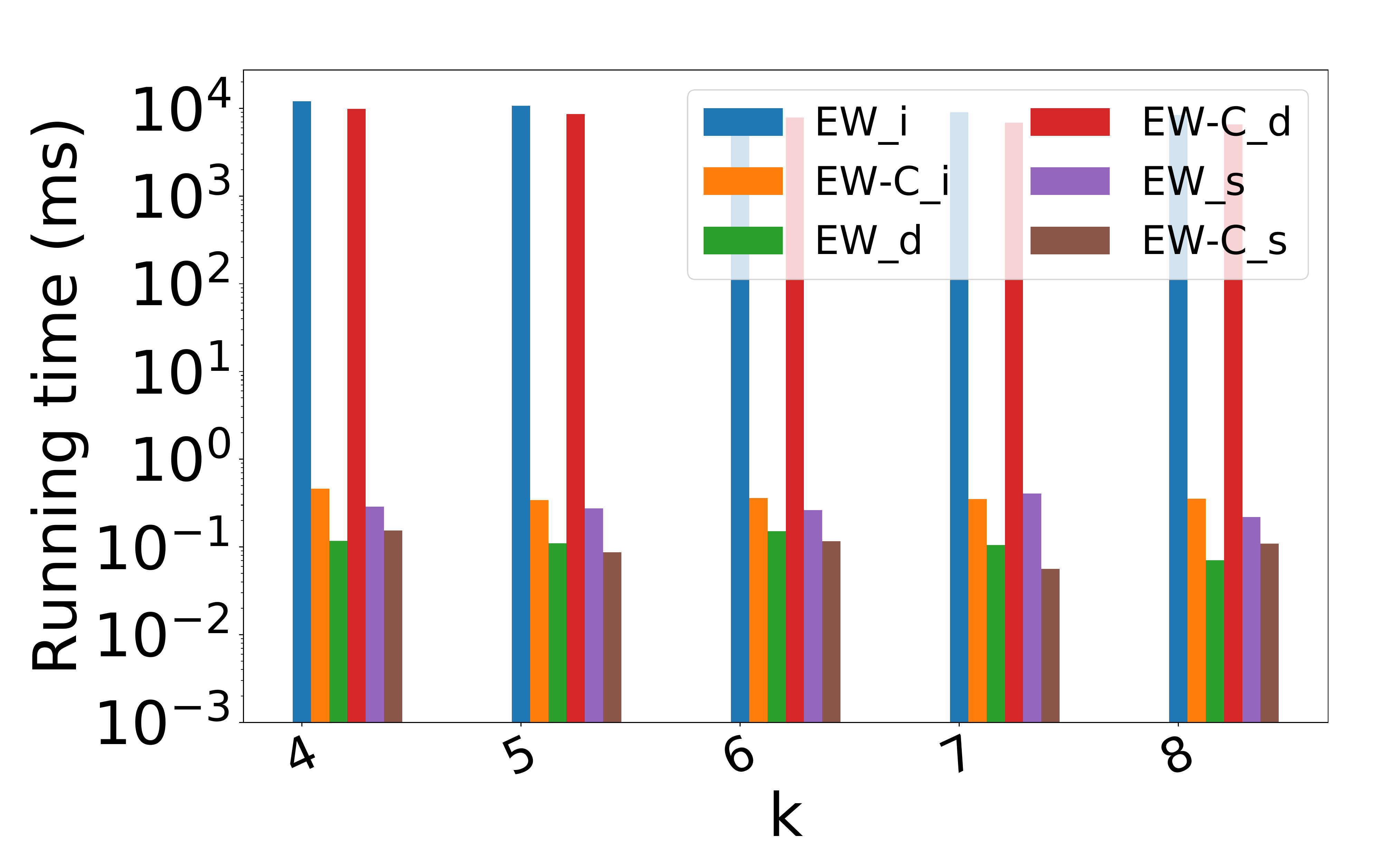}
\subcaption{Record\_label}
\end{subfigure}
\hspace{2em}
\begin{subfigure}{0.2\textwidth}
\includegraphics[scale=0.11]{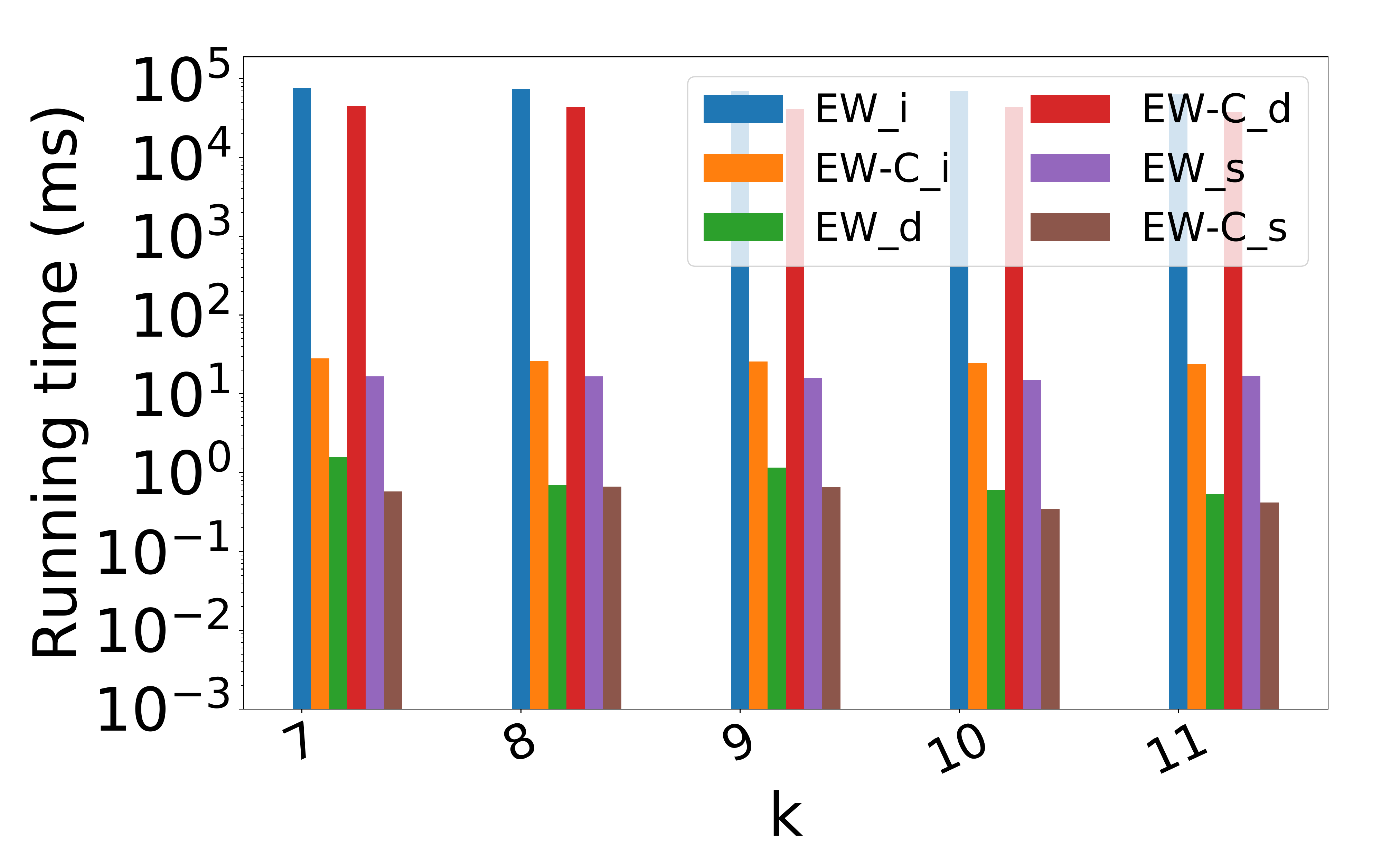}
\subcaption{YouTube}
\end{subfigure}
\hspace{2em}
\begin{subfigure}{0.2\textwidth}
\includegraphics[scale=0.11]{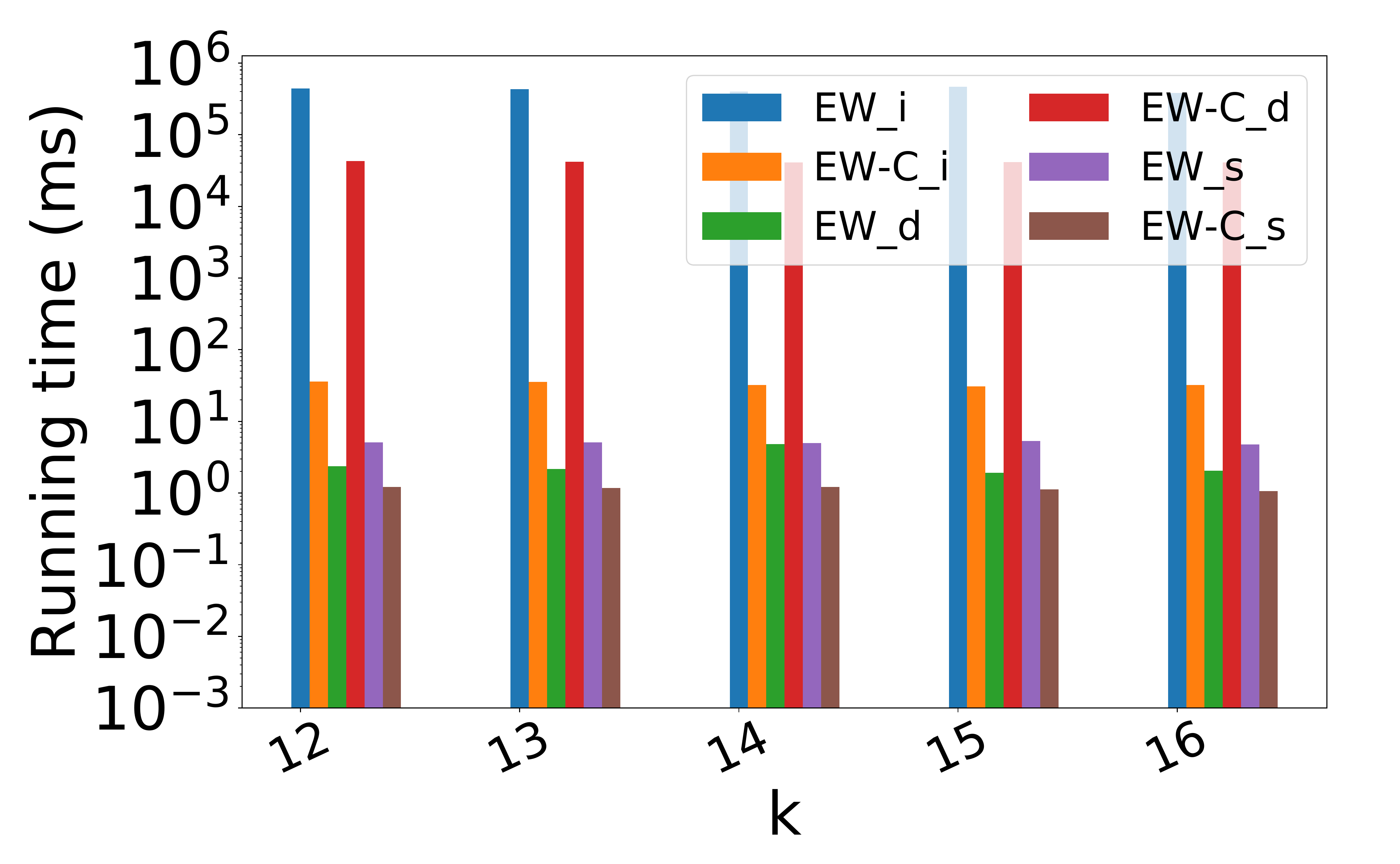}
\subcaption{Stackoverflow}
\end{subfigure}
\caption{Cohesive subgraph performance on different wing number values $k$ on real-world datasets}
\label{fig:exp_k}
\vspace{-12pt}
\end{figure*}
After the indices are constructed we can now utilize them to perform a $k$-wing search for a query vertex $q$. 
The experiment settings for evaluating query processing efficiency are inspired by \cite{huang2014querying, akbas2017truss}. We consider two experimental settings below. 

%TODO: graph data management, spatial data management

In the first set of experiments, to evaluate the diverse variety of the query processing time, vertices with different degrees across the datasets are selected randomly as query vertices. 
For each dataset, we sort all the vertices $v \in U\cup V$ in non-increasing order of their degrees and segregate them equally into ten buckets. The ordering is such that the first bucket contains the top 10\% high-degree vertices and the last one contains the last 10\% low-degree vertices in $G$. Once the ordering is obtained, we now perform a random selection of 100 query vertices from each bucket and the average query processing time for each bucket is reported in Fig. \ref{fig:exp}. The corresponding $k$ values for each dataset are $k=4$ for $Producer$, $k=5$ for $Record\_label$, $k=10$ for $YouTube$, and $k=13$ for $Stackoverflow$. 

We report the following experimental observations in different real-world datasets. (1) \textit{BaseLine} approach is the least efficient algorithm. It is very much slower than the index-based approaches $EW$ and $EW\raisebox{0.5pt}{-}{C}$. In $Stackoverflow$, the queries take approximately $20$ minutes for \textit{BaseLine}, whereas the index-based approaches require less than a second. Hence, it is omitted in Fig. \ref{subfig:stack_query} for the $Stackoverflow$ dataset. This vast difference is due to the BFS exploration of the bipartite graph and the costly butterfly-connectivity evaluations in the \textit{BaseLine} approach. Therefore, processing a personalized $k$-wing query without an index becomes impractical in the real-world datasets. (2) In all the datasets, the $k$-wing search time drops firmly as the queries drawn from high-degree percentile buckets to low-degree buckets specifically for $Producer$, $Record\_label$ and $Stackoverflow$. While for the $YouTube$, the $k$-wing search time remains almost constant. (3) For the smaller datasets $Producer$ and $Record\_label$ the maximum difference in the performances of the indices is at the sixth bucket, which itself is less than a millisecond, and hence the running time of the two index-based approaches can be considered to be comparable on these datasets. However, as the input dataset size increases $EW$ becomes significantly slower, $EW\raisebox{0.5pt}{-}{C}$ is at least $10$ times faster than $EW$ for high-degree query vertices, and at least twice the order of $EW$ for the low-degree query vertices. %(4) Although \textit{EW} requires an auxiliary structure \textit{H}, \textit{EW-C} proved to be faster than \textit{EW} in larger datasets implying the computation of $seed(q)$ becomes negligible. This experiment establishes the superiority of using \textit{EquiWing-Comp} index. 

\begin{table}[t]
\caption{Index characteristics }
\centering
\scalebox{0.63}{
%\Huge
\begin{tabular}{|l|l|l|l|l|l|l|l|l|}
\hline
\multirow{2}{*}{Graph} & Graph  & \multicolumn{2}{l|}{$\#$ Super Nodes} &\multicolumn{2}{l|}{\begin{tabular}{@{}l@{}}Index \\ Space (MB)\end{tabular}} &
{\begin{tabular}{@{}l@{}}Compre \\ -ssion \end{tabular}} &   
\multicolumn{2}{l|}{\begin{tabular}{@{}l@{}}Construction \\ Time (sec.)\end{tabular}} \\ \cline{3-6}\cline{8-9} 
    &  Size (MB)    & $EW$   & $EW\raisebox{0.5pt}{-}{C}$   & $EW$  & $EW\raisebox{0.5pt}{-}{C}$   & Ratio & $EW$ & $EW\raisebox{0.5pt}{-}{C}$       \\ \hline
Producer& 1.09 & 5323 & 2986  & 1.70  & 0.89   & 1.78 & 21.61 & 24.72 \\ \hline
Record\_label& 1.3 & 1965 & 392 & 2.01  & 0.92 & 5.01 & 400.78 & 400.83 \\ \hline
YouTube & 2.7 & 31001 & 609 & 19.84  & 2.47 & 50.90 & 1397.10 & 1410.62 \\ \hline
Stackoverflow & 11.4 & 93275 & 790 & 57.51  & 8.85 & 118.07 & 14397.16 & 14457.46  \\ \hline
%graph size after removing extra information
%     &  Size (MB)    & EW   & EW-C   & EW (H) & EW-C   & Ratio & EW & EW-C       \\ \hline
% Producer&2.5& 5323 & 2986  & 1.70 (0.41) & 0.89   & 1.78 & 21.61 & 24.72 \\ \hline
% Record\_label& 2.8 & 1965 & 392 & 2.01 (0.51) & 0.92 & 5.01 & 400.78 & 400.83 \\ \hline
% YouTube&3.6&31001&609 & 19.84 (0.78) & 2.47 & 50.90 & 1397.10 & 1410.62 \\ \hline
% Stackoverflow&33& 93275 & 790 & 60.29 (2.78) & 8.85 & 118.07 & 14397.16 & 14457.46  \\ \hline
\end{tabular}
}
\label{table:index_construction}
\vspace{-10pt}
\end{table}

For the second set of experiments, we analyze the running time for the $k$-wing search across different real-world datasets by varying the parameter $k$. 
For each value of $k$ in the experiments, we form two query sets: the first set contains 100 high-degree vertices (selected at random from the first 30\% of the sorted vertices); the second set contains 100 low-degree vertices (selected at random from the remaining 70\% vertices). We evaluate all the algorithms for each set of queries denoted with the suffix $\_H$ (High) and $\_L$ (Low). Henceforth, we have $EW\_H$ and $EW\_L$ for \textit{EquiWing} index, $EW\raisebox{0.5pt}{-}{C}\_H$ and $EW\raisebox{0.5pt}{-}{C}\_L$ for \textit{EquiWing-Comp} index, and \textit{BaseLine\_H} and \textit{BaseLine\_L} for \textit{BaseLine} algorithm, for the corresponding query sets. The average running time for each method is reported in Fig. \ref{fig:exp_k} while varying the parameter $k$. We observe that for most of the datasets, $EW\raisebox{0.5pt}{-}{C}$ is the most efficient $k$-wing search method, which is at least one order of magnitude faster than $EW$ for the high-degree queries, for maximum values of $k$. 
% The main reasons are, in \textit{EW} index, there is unnecessary access to the super nodes which have the $\psi < k$, and also a large number of super nodes in \textit{EW}, which again increases the time for performing BFS in the index. In \textit{EW-C}, however, the hierarchical structure reduces unnecessary access and a reduced number of compressed super nodes also expedites for better query processing.
The better performance of $EW\raisebox{0.5pt}{-}{C}$ can be explained using the reduced number of super nodes in comparison to $EW$, which expedites the query processing. 
The performance gap increases significantly with the larger datasets, such as $YouTube$ and $Stackoverflow$. These experimental evaluations yet again provide the superiority of \textit{EquiWing-Comp} for $k$-wing search. Moreover, we can easily infer that an approach without an index is infeasible and impractical in large datasets. For $Stackoverflow$, the running time of \textit{BaseLine} is at least  $10^4$ times slower than that of the algorithms using the indices.
\vspace{-5pt}

\subsection{Dynamic maintenance of $EW$ and $EW\raisebox{0.5pt}{-}{C}$}
\begin{table}[]
\centering 
\caption{Dynamic maintenance (sec.) of $EW$ and $EW\raisebox{0.5pt}{-}{C}$}%EquiWing and EquiWing-Comp} %for edge insertion/deletion
\scalebox{0.71}{
%\Huge
\begin{tabular}{|l|l|l|l|l|l|l|}
\hline
\multirow{2}{*}{\centering Graph} & \multicolumn{2}{l|}{ \centering Insertion} &\multicolumn{2}{l|}{\centering Deletion} &\multicolumn{2}{l|}{\centering Computing index from scratch} \\ \cline{2-7}
& $EW$   & $EW\raisebox{0.5pt}{-}{C}$   & $EW$ & $EW\raisebox{0.5pt}{-}{C}$ & $EW$ & $EW\raisebox{0.5pt}{-}{C}$  \\  \hline
% Producer &0.009 & 0.009  & 0.098  & 0.110 & 21.61 & 24.72  \\ \hline
% Record\_label & 0.037 & 0.037  & 28.08 & 28.11 & 400.78  & 400.83\\ \hline
% YouTube & 0.012 & 0.012  & 31.99 & 34.67 & 1397.10 & 1410.62 \\ \hline
% Stackoverflow & 328.23 & 328.23  & 1937.41 & 1951.42 & 14397.16  & 14457.46 \\ \hline
Producer &0.009 & 0.009  & 0.162  & 0.529 & 21.61 & 24.72  \\ \hline
Record\_label & 0.076 & 0.165  & 17.307 & 18.298 & 400.78  & 400.83\\ \hline
YouTube & 0.061 & 4.214 & 20.347 & 176.071 & 1397.10 & 1410.62 \\ \hline
Stackoverflow & 395.815 & 452.563  & 1761.93 & 3464.93 & 14397.16  & 14457.46 \\ \hline
\end{tabular}
}
\label{table:dynamic}
\vspace{-10pt}
\end{table}

%\begin{table}[]
%\centering 
%\caption{ Dynamic maintenance (seconds) of EquiWing and
%EquiWing-Comp for edge insertion/deletion }
%\scalebox{0.78}{
%\Huge
%\begin{tabular}{|l|l|l|l|l|l|l|}
%\hline
%\multirow{2}{*}{Graph} & \multicolumn{2}{l|}{ Insertion} &\multicolumn{2}{l|}{ Deletion} &\multicolumn{2}{l|}{ Computing from Scratch} \\ \cline{2-7}
%& EW   & EW-C   & EW & EW-C & EW & EW-C  \\  \hline
%Producer &0.009 & 0.009  & 0.098  & 0.110 (0.012) & 21.61 & 24.72  \\ \hline
%    Record\_label & 0.037 & 0.037  & 28.08 & 28.11 (0.032) & 400.78  & 400.83\\ \hline
%    YouTube & 0.012 & 0.012  & 31.99 & 34.67 (2.68) & 1397.10 & 1410.62 \\ \hline
 %   Stackoverflow & 328.23 & 328.23  & 2380.34 & 2394.78 %(14.44) & 14397.16  & 14457.46  \\ \hline
%\end{tabular}
%}
%\label{table:dynamic}
%\vspace{-10pt}
%\end{table}
\vspace{-5pt}
In this set of experiments, we compare the performances of incremental update of the $EW$ and $EW\raisebox{0.5pt}{-}{C}$ when $G$ is updated with new edges insertion and existing edges deletion. We randomly select 100 edges for insertion/deletion, and update the wing number for the edges in both the indices after each insertion/deletion. The average update time which includes the edge wing number update and the indices update time, is reported in Table \ref{table:dynamic}. $EW$ is updated using Algorithm \ref{Algorithm:dynamic_update}, further we compress the updated $EW$ to produce $EW\raisebox{0.5pt}{-}{C}$ with a marginal extra cost for compression. The compression cost for $EW\raisebox{0.5pt}{-}{C}$ is almost trivial for insertion in all the datasets while it is comparatively more for the deletion yet the total update time for $EW\raisebox{0.5pt}{-}{C}$ is still very less w.r.t construction of $EW\raisebox{0.5pt}{-}{C}$ from the scratch.
%yet it is still negligible w.r.t $EW$ update time. 
We have also reported the index construction time for the $EW$ and $EW\raisebox{0.5pt}{-}{C}$ indices from scratch.% when $G$ is updated with an edge insertion/deletion. 

The results in Table \ref{table:dynamic} show that the update time per edge insertion ranges from $0.04\%$ (Producer) to $2.3\%$ (Stackoverflow) of the index construction from scratch. Thus, handling edge insertion is highly efficient. For the deletion case, the update time per edge deletion 
%is at least seven times better than 
ranges from $0.4\%$ (Producer) to $23.9\%$ (Stackoverflow) of
%$13.5\%$ (Stackoverflow) of
the index construction from scratch. We can see that the incremental update approaches are several orders of magnitude faster than constructing the $EW$ and $EW\raisebox{0.5pt}{-}{C}$ from scratch when $G$ is updated. This demonstrates the superiority of our proposed incremental update algorithms.

\subsection{Case Study on Unicode}
We conduct a case study on the \textit{Unicode} dataset \cite{kunegis2013konect} to exhibit the effectiveness of \textit{personalized $k$-wing model}. It is composed of two sets of vertices i.e. languages and countries. An edge between two vertices depicts the language spoken in the country. We perform the $k$-wing search for the English language ($’en’$) in the \textit{Unicode} dataset with $k=$6 and $k=7$. The results are shown in Fig. \ref{subfig:case2} and \ref{subfig:case3}, respectively. We recognize that $’en’$ participates with two \textit{6-wings}. The first one contains edges with the maximum $\psi=6$ (green). The second one contains edges with $\psi=14$ (blue), which indicates it can further contain more dense $k$-wings. Therefore, in Fig. \ref{subfig:case3}, we set the value of $k=7$, the resulting $k$-wings are formed with the vertices (languages) $'fr','ar','es','tr','de'$ and $'el'$. Fig. \ref{subfig:case3} shows the green subgraphs are dissolved indicating that $'en'$ is more tightly related to the other $k$-wings. %The thickness of the edge represents the $k$ values (the higher the value of $k$, the more is the thickness). %
Hence, by tuning $k$, we can query a number of $k$-wings with different density and cohesiveness, which is crucial for personalized $k$-wing search in large graph analysis and studies.
Note that we duplicate some languages that participate in more than one $k$-wings in Fig. \ref{fig:casestudy}, e.g., $'ku\_1 '$ and $'ku\_2 '$, for a better visual representation. 
%The languages obtained from the $k$-wings of $'en'$ can be used to improve the communication among countries.% by providing a common platform.% We can also observe the extent of cohesiveness (distinct values of $\psi$) of other languages with $'en'$ in Fig. \ref{subfig:case1}.

\section{Related Work}
\begin{figure}[]

\begin{subfigure}{0.23\textwidth}
\includegraphics[scale=0.16]{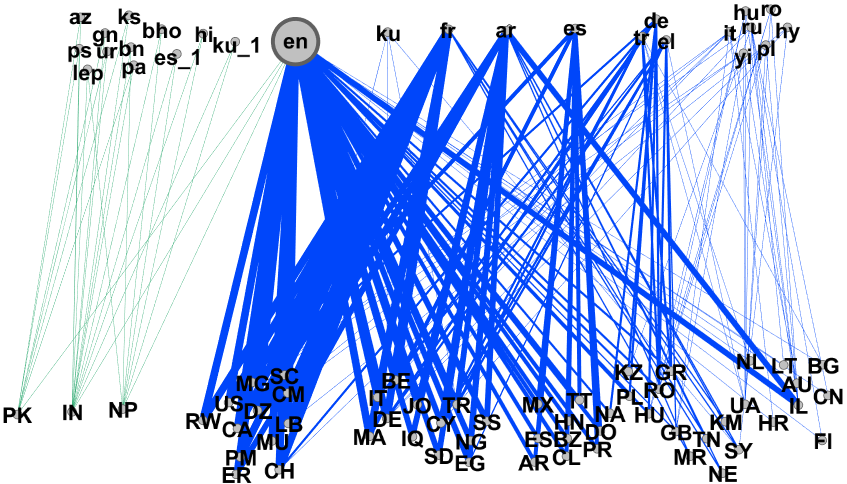}
\subcaption{Two 6-wings for $q=en$.}
\label{subfig:case2}
\end{subfigure}
\hspace{1em}
\begin{subfigure}{0.23\textwidth}
\centering
\includegraphics[scale=0.125]{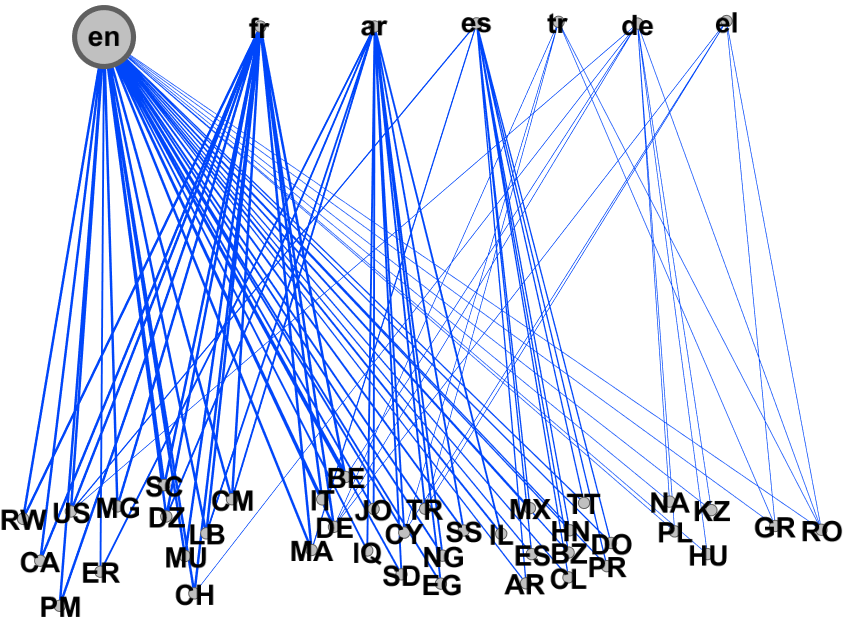}
\subcaption{One 7-wing for $q=en$.}
\label{subfig:case3}
\end{subfigure}
% \begin{subfigure}{0.34\textwidth}
% \includegraphics[scale=0.2]{All_wing_en.png}
% \subcaption{1-wings containing \textit{en}, with $k_{max}=14$.}
% \label{subfig:case1}
% \end{subfigure}
% \hspace{2.5em}

\caption{$k$-wings in Unicode Language datasets}
%\caption{For the Unicode Language datasets, we query for English language with different $k$ values. }
\label{fig:casestudy}
\vspace{-10pt}
\end{figure}

% \begin{figure*}[]
% \hspace{-2em}
% \begin{subfigure}{0.23\textwidth}
% \includegraphics[scale=0.15]{q=en,k=6.png}
% \subcaption{Two 6-wings for $q=en$.}
% \label{subfig:case2}
% \end{subfigure}
% \hspace{2.5em}
% \begin{subfigure}{0.23\textwidth}
% \centering
% \includegraphics[scale=0.12]{q=en,k=7.png}
% \subcaption{One 7-wing for $q=en$.}
% \label{subfig:case3}
% \end{subfigure}
% \begin{subfigure}{0.34\textwidth}
% \includegraphics[scale=0.17]{All_wing_en.png}
% \subcaption{1-wings containing \textit{en}, with $k_{max}=14$.}
% \label{subfig:case1}
% \end{subfigure}
% \hspace{2.5em}
% \caption{For the Unicode Language datasets, we query for English language with different $k$ values. }
% \label{fig:casestudy}
% \vspace{-10pt}
% \end{figure*}

\label{section:related_work}
%The existing works in cohesive subgraph retrieval can be divided into the following sets. \\
\noindent\textbf{Bipartite cohesive subgraph detection.} In this set of work, the problem focuses on enumerating all the cohesive subgraphs in a bipartite graph. Recently, the bipartite graph decomposition has attracted a lot of attention from the researchers \cite{DBLP:conf/wsdm/SariyuceP18, zou2016bitruss, DBLP:conf/icde/Wang0Q0020}, which can be used for querying cohesive subgraphs. Even though \cite{Sanei-Mehri:2018:BCB:3219819.3220097} and \cite{DBLP:journals/pvldb/WangLQZZ19} improved the decomposition algorithm by improving the butterfly counting, yet they are inefficient for repetitive query processing. \cite{liu2020efficient} presented an index-based approach to enumerate all the $(\alpha,\beta)$-core structures in a bipartite graph. The $(k,P)$-Bitruss model for the bipartite graph was also proposed in \cite{yang2020effective} to determine the densely connected vertex of the same type. Recently, \cite{fang2020effective} proposed a \textit{community search (or cohesive subgraph search)}  in heterogeneous networks based on \textit{$(k,P)$-core} structure.
%, where $P$ is the meta path and $k$ is the minimum degree of a vertex, connected via meta path in the \textit{community}. 
The dynamic maintenance of the index for \textit{cohesive subgraph search} was discussed in \cite{liu2020efficient}.\\
\noindent\textbf{Cohesive subgraph search.}
The cohesive subgraph search has been implemented using different distinct cohesive structures such as \textit{$k$-core} \cite{Sozio2010cocktail}, \textit{$k$-truss} \cite{huang2014querying,akbas2017truss}, and \textit{$k$-edge connected subgraph} \cite{zhou2012finding,burt2005brokerage,chang2013efficiently}, etc. However, all of them are from the field of unipartite graphs. Unfortunately, there are very few such cohesive subgraph search models for bipartite graphs. 
Recently, \cite{wang2020efficient} proposed a community search in a weighted bipartite graphs using a $(\alpha,\beta)$-core model. The resulting significant $(\alpha,\beta)$-community $\mathcal{R}$ for a query vertex $q$ adopts $(\alpha,\beta)$-core to characterize the
engagement level of vertices, and maximizes the minimum edge weight (significance) within $\mathcal{R}$.
%Recently, \cite{fang2020effective} proposed a community search in heterogeneous networks based on \textit{$(k,P)$-core} structure. The resulting community search requires an extra parameter  $P$ and it results in the weak cohesive community as the length of the path increases. Moreover, the resulting community only includes vertices of the same type rather than the edges from the graph, hence an overhead occurs if we want to inquire about other interests. 
\cite{akbas2017truss} and \cite{huang2014querying} also discussed the incremental dynamic maintenance algorithm for $k$-truss.%The dynamic maintenance algorithms for core number in graphs are proposed in \cite{sariyuce2016incremental, zhang2017fast}.

\vspace{-5pt}
\section{Conclusion}
\label{section:conclusion}
% In this paper, we examined the problem of personalized $k$-wing search for large and dynamic bipartite graphs. We proposed two indices \textit{EquiWing} and \textit{EquiWing-Comp} to tackle the problem efficiently, which led to linear time search algorithms. We constructed \textit{EquiWing-Comp} by proposing the $k$-butterfly loose connectivity and exploiting the hierarchical property of the $k$-wing, which further sped up the query processing. We also discussed the efficient maintenance of the proposed indices in dynamic bipartite graphs. We then conducted extensive experiments across large real-world datasets. From these experiments, we observed the superiority of our index based approaches \textit{EquiWing} and \textit{EquiWing-Comp} over the baseline approach. The experiments also displayed the better compression and performance of \textit{EquiWing-Comp} over \textit{EquiWing}. Moreover, a case study has been presented to display the effectiveness of our personalized $k$-wing model.
In this paper, we examine the problem of personalized $k$-wing search for large and dynamic bipartite graphs. 
We propose two indices \textit{EquiWing} and \textit{EquiWing-Comp} to tackle the problem efficiently, which leads to linear time search algorithms. 
%The \textit{$k$-butterfly equivalence} successfully summarized the bipartite graph into \textit{EquiWing} without losing any edges. 
We construct \textit{EquiWing-Comp} by proposing the \textit{$k$-butterfly loose connectivity} and exploiting the hierarchical property of the $k$-wing, which further speeds up the query processing. 
We also discuss the efficient maintenance of the proposed indices in dynamic bipartite graphs.
We conduct extensive experiments across large real-world datasets. 
From these experiments, we observe the superiority of our index based approaches \textit{EquiWing} and \textit{EquiWing-Comp} over the baseline approach. The experiments also display the better compression and performance of \textit{EquiWing-Comp} over \textit{EquiWing}. 
Moreover, a case study is presented to display the effectiveness of our personalized $k$-wing model.

\renewcommand{\IEEEbibitemsep}{0pt}
\makeatletter
\IEEEtriggercmd{\reset@font\normalfont\footnotesize}
\makeatother
\IEEEtriggeratref{1}
\bibliographystyle{abbrv}
\bibliography{00main}
\end{document}